\documentclass[a4paper,twocolumn,superscriptaddress,11pt, accepted=2020-02-20
]{quantumarticle} %
\pdfoutput=1

\usepackage{nicefrac}
\usepackage[english]{babel}
\usepackage{soul} 
\usepackage{array} 
\usepackage{hyperref}
\usepackage{xcolor}
\definecolor{mylinkcolor}{rgb}{0,0,0.8} 
\hypersetup{unicode=true, %
	bookmarksnumbered=false,bookmarksopen=false,bookmarksopenlevel=1, %
	breaklinks=true,pdfborder={0 0 0},colorlinks=true}%
\hypersetup{%
	anchorcolor=mylinkcolor,citecolor=mylinkcolor, %
	filecolor=mylinkcolor,linkcolor=mylinkcolor, %
	menucolor=mylinkcolor,runcolor=mylinkcolor, %
	urlcolor=mylinkcolor}%

\usepackage{bbm} 
\usepackage{amsmath,amsfonts}
\usepackage{graphicx}
\usepackage{amsthm}
\usepackage{enumerate}
\usepackage{tikz}
\usepackage{bookmark}

\usepackage[numbers,sort&compress]{natbib}
\makeatletter
\def\pdfstartlink@attr{attr{/Border[0 0 0 [1 5] ]/H/I/C[0 1 1]}}%
\def\@@Doi#1{\textcolor{mylinkcolor}{#1}\@@endlink}
\makeatother

\newtheorem{thm}{Theorem} 
\newtheorem{prop}{Proposition} 
\newtheorem{pty}{Property} 
\newtheorem{lem}{Lemma} 
\newtheorem{cor}{Corollary} 
\theoremstyle{definition}

\newtheorem*{rem*}{Remark} 

\newcommand{\N}{\mathbb{N}}

\newcommand{\R}{\mathbb{R}}

\newcommand{\ket}[1]{| #1 \rangle}

\newcommand{\ketbra}[2]{|#1\rangle\!\langle#2|}
\newcommand{\ot}{\otimes}

\newcommand{\stah}{ \ \mathrm{s.t.} \ }

\newcommand{\hplo}{H_{+}}
\newcommand{\hpl}[1]{H_{+}\!\left(#1 \right)}

\newcommand{\rC}{\mathrm{C}}
\newcommand{\rQ}{\mathrm{Q}}
\newcommand{\rB}{\mathrm{B}}
\newcommand{\rG}{\mathrm{G}}

\newcommand{\sta}{\mathcal{S}}
\newcommand{\eff}{\mathcal{F}}
\newcommand{\me}{\mathcal{E}}
\newcommand{\esp}{\mathcal{R}}
\newcommand{\tra}{\mathcal{T}}

\begin{document}
	\title{Analysing causal structures in generalised probabilistic theories}
	\author{Mirjam Weilenmann}
	\email{mirjam.weilenmann@oeaw.ac.at}
	\affiliation{Institute for Quantum Optics and Quantum Information (IQOQI)
		Vienna, Austrian Academy of Sciences, Boltzmanngasse 3, 1090
		Vienna, AT}
	\affiliation{Department of Mathematics, University of York, Heslington, York, YO10 5DD, UK}
	\orcid{0000-0003-0852-6763}
        \author{Roger Colbeck}
	\email{roger.colbeck@york.ac.uk}
	\affiliation{Department of Mathematics, University of York, Heslington, York, YO10 5DD, UK}
	\orcid{0000-0003-3591-0576}

	\begin{abstract}
		Causal structures give us a way to understand the origin of observed correlations. These were developed for classical scenarios, but quantum mechanical experiments necessitate their generalisation. Here we study causal structures in a broad range of theories, which include both quantum and classical theory as special cases.  We propose a method for analysing differences between such theories based on the so-called measurement entropy. We apply this method to several causal structures, deriving new relations that separate classical, quantum and more general theories within these causal structures. The constraints we derive for the most general theories are in a sense minimal requirements of any causal explanation in these scenarios. In addition, we make several technical contributions that give insight for the entropic analysis of quantum causal structures. In particular, we prove that for any causal structure and for any generalised probabilistic theory, the set of achievable entropy vectors form a convex cone.
	\end{abstract}
	
	\maketitle
	
	Given a set of observed variables, some of which may be correlated, a causal structure gives a more detailed picture of how the correlations come about.  Depending on the situation, this causal structure may posit the existence of hidden common causes and the nature of these depends on the physical theory.  For instance, the experimental violation of a Bell inequality~\cite{Bell1964} can be explained either by adapting the causal structure within the realm of classical physics (at the expense of resorting to fine-tuning~\cite{Wood2012}) or by allowing hidden systems to be non-classical.
	
	Causal structures also provide a suitable basis for analysing the features of different theories by allowing us to phrase communication and cryptographic protocols in terms of the dependencies among the involved systems. They help us predict the success of players engaged in a protocol when restricted according to different theories, for example, in random access coding and the related principle of Information Causality~\cite{Pawlowski2009, Al-Safi2011}.
	
	The differences between the observable correlations that can be achieved with classical and quantum resources within a given causal structure have been extensively analysed, starting with the derivation of several classical constraints and their quantum violations~\cite{Clauser1969, GHZ}, and progressing to a systematic analysis~\cite{Fritz2012,Fritz2013,Chaves2015,nonshan,Wolfe2016}.  Less work has been dedicated to understanding the limitations of quantum systems~\cite{Tsirelson1980, Chaves2015, VanHimbeeck2019} and of the behaviour of theories beyond.  For the latter, there have been analyses of the implications of the no-signalling principle on causal structures~\cite{Henson2014, Chaves2016}. More generally, understanding the differences of generalised probabilistic theories (GPTs) with respect to different tasks may inform the search for principles that single out quantum mechanics.
	
	In this work, we introduce a technique for deriving constraints on the observable correlations that are achievable in different causal structures according to different GPTs, with the aim of moving towards a systematic analysis of the differences between such theories.  Our approach is based on \emph{measurement entropy}~\cite{Short2010, Barnum2010} and inspired by entropic approaches to analysing causal structures involving classical and quantum resources~\cite{Braunstein1988, Steudel2015, Chaves2012, Fritz2013, Chaves2015, Pienaar2016, review}.  That such a generalisation is possible, was not at all clear, since work regarding the definition of entropy in GPTs showed that there is no entropy measure that retains the relevant properties of the von Neumann entropy~\cite{Short2010, Barnum2010}. In particular, the additivity of entropy under the composition of different systems is not retained by the proposed measures, which in previous entropic approaches for analysing causal structures was crucial for encoding causal constraints.
	
	One of the key points that allows us to overcome these issues is to explicitly include the conditional entropy in the analysis. Nevertheless, since our final results are stated in terms of the Shannon entropy they can be directly compared to those obtained with previous entropic techniques.
	
	We apply our method to various causal structures, generating a series of entropic constraints that exclude certain causal explanations of observed correlations when restricted by arbitrary GPTs. This allows us to compare different causal structures with respect to the correlations they allow in different theories (in particular, we compare classical, quantum and arbitrary GPTs).  In some cases we find the same sets of entropic description regardless of the theory (here known quantum inequalities also apply to GPTs), while in others we can show an entropic separation.  For instance, we apply our technique to Information Causality~\cite{Pawlowski2009}, a candidate principle for singling out quantum theory, showing that our method improves upon that of~\cite{Chaves2016}, yielding the stronger inequalities of~\cite{Short2010a}. Although the maximally non-local GPT, \emph{box-world} does not satisfy the notion of Information Causality, we identify minimal notions of causation that are satisfied.
	
	In addition to providing a method for analysing causal structures with GPT resources, we make technical contributions by showing that any set of achievable entropy vectors for the observed variables in a causal structure involving quantum or other generalised probabilistic resources is a convex cone. Previously this had only been shown for the entropy vectors of classical resources~\cite{Zhang1997, Chaves2012, Fritz2013, review}. This insight allows for easy comparison of the entropic sets within different theories, and in some cases enables us to prove that a given characterisation is complete by showing that all extremal points are achievable.  We also give some insights into the entropic analysis of quantum causal structures.

	\section{Preliminaries}
	For every system $A$ in a GPT, there is an associated state space $\sta_A$, a compact convex subset of a real vector space $V$ and an associated space of effects, $\eff_A$. An \emph{effect} $e\in\eff_A$ is a linear map $\sta_A\to[0,1]$ (thus, $e$ is a vector in the dual space to $V$). There is a special effect, $u_A\in\eff_A$, called the unit effect, with the property that $u_A(s)=1$ for all $s\in\sta_A$. A \emph{measurement} $M$ is a collection of effects whose sum is the unit effect, i.e., we can write $M=\{e^x\in\eff_A:\sum_xe^x=u_A\}$.  We use $\me_A$ to represent the set of allowed measurements on $A$. The interpretation of $e^x(s)$ is the probability of outcome $x$ when $M$ is performed on a system in state $s$.
	
	Consider two measurements on $A$: $M=\{e^x\}_{x\in\esp_M}$ and $N=\{f^y\}_{y\in\esp_N}$. If there exists a map $F:\esp_M\to\esp_N$ such that
	\begin{equation} \label{eq:finegrain}
	\sum_{x\in\esp_M:F(x)=y}e^x=f^{y}\ \forall y\in\esp_N
	\end{equation}
	we say that $M$ is a \emph{refinement} of $N$ (equivalently, $N$ is a \emph{coarse-graining} of $M$).\footnote{If $F$ is bijective,~\eqref{eq:finegrain} is a \emph{relabelling}. The set of measurements $\me_A$ is assumed to be closed under relabelling and coarse-graining.} A refinement is \emph{trivial} if for all $x\in\esp_M$, $e^x=c(x)\,f^{F(x)}$ for some $c(x)\in\mathbb{R}_{>0}$. The subset of \emph{fine-grained measurements}, $\me^*_A$, are those for which there are no non-trivial refinements. Throughout this article we restrict to GPTs where there is at least one finite-outcome fine-grained measurement (in classical and quantum theory this is a restriction to finite-dimensional systems).
	
	Transformations of systems are represented by linear maps between state spaces, $T:\sta_A\to\sta_B$ and the set of such transformations is denoted $\tra_{A\to B}$. The set $\tra_{A\to A}$, contains the identity transformation, $I_A$, and is closed under composition. Furthermore, a transformation followed by a measurement is a valid measurement.
	
	Two systems $A$ and $B$ can be thought of as parts of a single joint system $AB$.  We do not specify precisely what the joint state space is, but a minimal requirement is that if $s_A\in\sta_A$ and $s_B\in\sta_B$ then $s_A\ot s_B\in\sta_{AB}$.  States that can be written as $s_A\ot s_B$ are called \emph{product} and convex combinations thereof are \emph{separable}.  Analogously, measurements $M=\{e^x\}_{x\in\esp_M}\in\me_A$ and $N=\{f^y\}_{y\in\esp_N}\in\me_B$ can be composed into $L=M\ot N\in\me_{AB}$ with outcome set $\esp_L=\esp_M\times\esp_N$ and effects $g^{x,y}=e^x\ot f^y$. Product effects act on product states according to ${(e^x \ot f^y)(s_A \ot s_B)}=e^x(s_A)f^y(s_B)$.
	
	This implies that we have non-signalling theories: Suppose $\{e_a^x\}_x\in\me_A$ and $\{e_b^y\}_y\in\me_B$ are measurements for $a=1,\ldots,n_a$ and $b=1,\ldots,n_b$, then, for example,
	\begin{align*}
	p(y|a,b)
	=\sum_x(e_a^x\ot f_b^y)(s_{AB})=(u_A\ot f_b^y)(s_{AB})\,,
	\end{align*}
	which is independent of $a$.
	
	We also assume that there are well-defined reduced states: $\forall s_{AB} \in \sta_{AB} \ \exists s_A \in \sta_A \stah \forall e \in \eff_A, e(s_A)=(e \otimes u_B) s_{AB}$.  The post-measurement state on $A$ after a measurement on $B$ with outcome $x$ is
	\begin{equation} \label{eq:condstate}
	s_{A|x}= \frac{(I_A \otimes e^x)(s_{AB})}{e^x(s_B)}.
	\end{equation}

	If the system $A$ is classical, then $\sta_A$ is a simplex and (up to relabelling) there is only one fine-grained measurement that is not a trivial refinement of another fine-grained measurement. We call this a \emph{standard classical measurement}. Note that classical systems can be represented in any GPT and composing them maintains separability.
	
	\emph{Box world}~\cite{Barrett2007} is the GPT in which the joint state space of several systems is in one-to-one correspondence with the set of no-signalling distributions amongst those systems, i.e., its state space in this sense the largest possible within the framework.

	\section{Measurement entropy and its properties}
	The approach to analysing causal structures that we use in this work is based on measurement entropy. In this section we introduce this and outline some of its properties.
	
	Measurement entropy was first introduced in~\cite{Short2010a, Barnum2010}; we follow the exposition of~\cite{Short2010a} here. The \emph{measurement entropy}, $\hplo$, is the minimal Shannon entropy of the outcome distribution after a fine-grained measurement, i.e., for $s_A \in \sta_A$,
	\begin{equation}
	\hpl{A}=\inf_{\{e^x\}\in\me_A^*}-\sum_xe^x(s_A)\log_2 e^x(s_A) .
	\end{equation}
	Several ways to define the conditional measurement entropy have been proposed~\cite{Short2010a, Barnum2010}, of which we use the following~\cite{Short2010a}.  For any state $s_{AB}\in\sta_{AB}$ with reduced state $s_B\in\sta_B$, the \emph{conditional measurement entropy} is
	\begin{equation}
	\hpl{A|B}=\inf_{\{f^y\} \in \me_B}\sum_y f^y(s_B)\hpl{A_{|y}}\,,
	\end{equation} 
	where $\hpl{A_{|y}}$ is the entropy of the state on $A$ after a measurement on $B$ with outcome $y$, $s_{A|y}$.  For classical systems these entropies coincide with the Shannon entropy, $H$.
	
	The measurement entropy satisfies a list of properties that are useful to this work. Some of these have previously been derived in~\cite{Short2010a, Barnum2010}, others are new to this work.  In the remainder of this section, $\sta_A$, $\sta_{AB}$ etc.\ refer to state spaces within an arbitrary GPT.  For the proofs of the first two properties we refer to~\cite{Short2010a}.
	
	\begin{pty}[Positivity~\cite{Short2010a}]\label{it:1}
		$\hpl{A} \geq 0$ for all ${s_A \in \sta_A}$ and $\hpl{A|B} \geq 0$ for all $s_{AB} \in \sta_{AB}$.
	\end{pty}
	
	\begin{pty}[Reduction to Shannon entropy~\cite{Short2010a, Barnum2010}]\label{it:2}
		Let $A$ and $B$ be classical systems and $s_{AB} \in \sta_{AB}$, then $\hpl{A}=H(A)$ and $\hpl{A|B}=H(A|B)$.
	\end{pty}
	
	\begin{pty}[Data processing]\label{it:dpi}
		Let $s_{AB}\in\sta_{AB}$, $T\in\tra_{B\to C}$ and $s_{AC}=(I_A\ot T)(s_{AB})$.  Then ${\hpl{A|B} \leq \hpl{A|C}}$.
	\end{pty}
	\begin{proof}
		If $\{f^j\}$ form a measurement on $C$, then $\{g^j\}$ form a measurement on $B$, where $g^j:s_B\mapsto f^j(T(s_B))$. It follows that
		\begin{align*}
		\hpl{A|C}&=\inf_{\{f^j\}\in\me_C}\sum_j f^j(s_C)\hpl{A_{|j}}\\
		&=\inf_{\{f^j\}\in\me_C}\sum_j g^j(s_B)\hpl{A_{|j}}\\
		&\geq\inf_{\{g^j\}\in\me_B}\sum_j g^j(s_B) \hpl{A_{|j}}\\
		&=\hpl{A|B}\qedhere
		\end{align*}
	\end{proof}

	\begin{pty}[Independence]\label{it:ind}
		If two systems $A$ and $B$ are independent, i.e., $s_{AB}=s_A \otimes s_B$, then $\hpl{A|B}=\hpl{A}$.
	\end{pty}
	\begin{proof}
		If $A$ and $B$ are independent, after any measurement $\{e^j\}\in\me_B$ on $B$ the post-measurement state on $A$ is $s_A$, independent of the outcome of the measurement. Therefore $\hpl{A|B}=\hpl{A}$.
	\end{proof}
	
	\begin{pty}[Classical subsystem inequalities]\label{it:cs}
		For a joint state $s_{ABC} \in \sta_{ABC}$ with classical subsystems $A$ and $B$, $\hpl{AB|C} \geq \hpl{A|C}$.
	\end{pty}
	\begin{proof}
		For any measurement $\{f^j\}\in \me_C$ we have
		\begin{align*}
		\sum_jf^j(s_C)H(AB)_{s_{AB}^j}&\geq\sum_jf^j(s_C)H(A)_{s_A^j}\\
		&\geq\hpl{A|C}\,.
		\end{align*}
		Applying this to a sequence of measurements that converge to $\hpl{AB|C}$ establishes the claimed result.
	\end{proof}

	\begin{pty}[Subadditivity~\cite{Short2010a, Barnum2010}]\label{it:b1}
		In GPTs for which $M\in\me^*_A$ and $N\in\me^*_B$ imply $M\ot N\in\me^*_{AB}$ (which holds for locally tomographic theories, including box-world) ${\hpl{A}+\hpl{B}\geq\hpl{AB}}$.
	\end{pty}
	
	We refer to~\cite{Short2010a} or Appendix~\ref{app:general} for a proof of Property~\ref{it:b1}. Note further that Property~\ref{it:b1} is the only one that does not hold in arbitrary GPTs.

	\begin{pty}[Lemma~C2 of~\cite{Short2010a}]\label{it:b2}
		For $s_{ABC} \in \sta_{ABC}$ where $A$ and $B$ are classical systems, $\hpl{A|BC} \geq \hpl{AB|C}-\hpl{B}$.
	\end{pty}
	
	We refer to~\cite{Short2010a} or Appendix~\ref{app:general} for a proof of Property~\ref{it:b2}.

	\begin{pty}\label{it:b4}
		For $s_{ABC} \in \sta_{ABC}$, where $B$ is a classical subsystem, ${\hpl{A|BC} \leq \hpl{AB|C} - \hpl{B|C}}$. If $C$ is also classical then this holds with equality.
	\end{pty}
	
	We prove this property in Appendix~\ref{app:general}. Note that Properties~\ref{it:b2} and~\ref{it:b4} are both relaxations of the chain rule, $H(A|BC)=H(AB|C)-H(B|C)$, that holds for Shannon and von Neumann entropy.
	
	\section{Entropy vector method for causal structures in GPTs}
	
	A \emph{causal structure} is a set of nodes arranged in a directed acyclic graph, some of which are labelled observed.  Each observed node has a corresponding random variable, while the other, unobserved nodes correspond to resources from a GPT. For a causal structure $C$ we use $C^\rC$, $C^\rQ$, $C^\rB$ or $C^\rG$ depending on whether the resources are classical, quantum, box-world systems or from some unspecified GPT respectively. For each unobserved node we associate a subsystem with each of its outgoing edges. An example for this is displayed in Figure~\ref{fig:instrumental}.
	
	A direct arrow from a node $A$ in a causal structure to a node $Z$ means that $A$ is a \emph{parent} of $Z$; a directed path from $A$ to $Z$ means that $A$ is an \emph{ancestor} of $Z$. For an unobserved node $A$, all subsystems associated with its outgoing edges are considered parents/ancestors of each its children/descendants. Given a causal structure, a \emph{coexisting set} of systems~\cite{Chaves2015,review} is one for which a joint state can be defined. In general, no coexisting set includes all nodes, since there is no joint state of a system and the output obtained from a measurement on it (unless the system is classical).
	
	Our method to generate new inequalities for causal structures with GPT resources begins by considering an \emph{entropy vector} whose components are the entropies and conditional entropies of all coexisting sets. Conditional entropies composed entirely of classical subsystems are excluded because they are linear combinations of other entropies (e.g., $H(X|Y)=H(XY)-H(Y)$).
	
	We then impose a system of linear (in)equalities that are necessary for a vector to be realisable as an entropy vector in a causal structure. These inequalities are constructed using the properties of the measurement entropy explained earlier and strong subadditivity in the cases where the measurement entropy reduces to the Shannon entropy.  In the case of locally-tomographic GPTs, such as box-world, there is one additional property (Property~\ref{it:b1}) that does not hold in all GPTs.  Further constraints come from the causal structure: two sets of nodes are independent if they do not share any ancestors in the causal structure. In general, there may be further independencies among the observed variables (see Theorem~22(i) of~\cite{Henson2014} and Appendix~\ref{app:quantum_entropy_vectors}).  This system of inequalities constrains a polyhedral cone, which can be projected to a \emph{marginal cone} that contains no components involving unobserved systems. The projection is performed with a Fourier-Motzkin elimination algorithm~\cite{Williams1986}. An example that illustrates this procedure in detail is provided at the beginning of Section~\ref{sec:applications1}.
	
	When dealing with causal structures for which computing all entropy inequalities for the marginal scenarios of interest is computationally impractical or even not possible with the computational resources at hand, due to the scaling of Fourier-Motzkin elimination~\cite{Monniaux2010}, we can still derive valid entropy inequalities by marginalising subsets of all valid inequalities.
	
	Furthermore, given a particular observed distribution that we suspect to be incompatible with a causal structure (either for classical theory, quantum theory or boxworld), it is not necessary to go through the marginalisation procedure discussed here. Instead we can look for a certificate of incompatibility using a linear program. This program can be set up by computing the entropy vector for the distribution in question and then adding this as a list of equalities (one for each of its components) to the list of valid entropy inequalities for the causal structure. If the resulting system of linear (in)equalities is infeasible, then the distribution in question is certified as incompatible with the causal structure within the theory under consideration.
	
	For some causal structures, the entropic constraints derived using Property~\ref{it:b1} are also valid for GPTs that are not locally tomographic, which is the content of the following proposition.
	
	\begin{prop} \label{prop:generalGPT} Let $C$ be a causal structure in which there are no nodes with two or more unobserved parents.  For any GPT $G$, any correlations achievable with a finite number of finite-outcome measurements in $C^\rG$ are achievable in $C^\rB$.
	\end{prop} 
	
	This proposition follows from the insight presented in the proof of Lemma~2 in~\cite{Cadney2012a} that any joint measurement on a classical and a GPT system can be written as a measurement on the classical system followed by an outcome-dependent measurement on the other (see also Lemma~\ref{lem:classfirst} in Appendix~\ref{app:general}).
	
	\begin{proof}
		Since each node has at most one unobserved parent, by Lemma~\ref{lem:classfirst} at each node we can assume a standard classical measurement on the classical subsystems followed by a measurement on the GPT subsystem depending on the result.  Consider then $s_{A_1A_2\ldots}\in\sta^\rG_{A_1A_2\ldots}$ and let $\{e^x_a\}_x\in\me_{A_1}$ for $a\in\{1,2,\ldots,m_a\}$, $\{f^y_b\}_y\in\me_{A_2}$ for $b\in\{1,2,\ldots,m_b\}$ etc. Since the outcome distribution $p(x,y,\ldots|a,b,\ldots)$ is no-signalling and since (by definition) all no-signalling distributions can be realised by states in box-world, there exists $s'_{A_1A_2\ldots}\in\sta^\rB_{A_1A_2\ldots}$ and box-world measurements on $A_1$, $A_2$, $\ldots$ that give rise to the same correlations.
	\end{proof}
	
	Note that the same argument does not hold if there are multiple unobserved parents at a single node.  This is because some joint measurements cannot be expressed as a measurement on one system followed by a measurement on the other conditioned on the first (cf.\ Lemma~\ref{lem:classfirst}), for example a measurement in the Bell basis in quantum mechanics.
	
	The method presented in this section recovers previous entropic approaches for describing classical~\cite{Chaves2012,Fritz2013} and quantum~\cite{Chaves2015} causal structures as special cases. In the classical case, the measurement entropy and its conditional version coincide with the Shannon entropy and all variables in a causal structure (observed and unobserved ones) coexist. In this case, our method is equivalent to that of~\cite{Chaves2012,Fritz2013}.\footnote{In this case Property~\ref{it:cs} imposes that the entropy is strongly subadditive.} For the quantum case, the recovery of the method proposed in~\cite{Chaves2015} from ours is less obvious and is explained in Section~\ref{sec:quantum}.
	
	When considering different causal structures, convexity of the sets of achievable entropy vectors of the observed variables is useful for their comparison: for instance, it allows us to prove that the achievable entropies in one case are contained in those of another by considering only the extreme points. The following theorem (proven in Appendix~\ref{app:generalGPT}) establishes convexity in general and is therefore an important structural insight.
	\begin{thm} \label{thm:convexity}
		For any causal structure $C^\rG$ the closure of the set of achievable entropy vectors of the observed variables is a convex cone.
	\end{thm}
	
	\subsection{Causal structures and post-selection}
	For a causal structure involving a parentless observed node $X$ that takes values $1,2,\ldots,n$, we can also analyse an adapted causal structure where each descendant of $X$ is split into $n$-variables and $X$ is dropped, e.g., a descendant $Y$ is split into $Y_{|X=1}, \ldots ,Y_{|X=n}$. The resulting causal structure is said to be \emph{post-selected} on $X$ (see Figure~\ref{fig:infcaus} for an example, and~\cite{Chaves2015,review} for further details of this procedure). For some causal structures, $C$, post-selection is necessary for deriving entropy inequalities that distinguish between $C^\rC$, $C^\rQ$ and $C^\rB$~\cite{Braunstein1988,Chaves2015,linelike}. When post-selecting on parentless nodes, convexity of the set of achievable entropy vectors of the observed variables also holds by the following corollary of Theorem~\ref{thm:convexity}, which is also proven in Appendix~\ref{app:generalGPT}.
	
	\begin{cor} \label{cor:convexity}
		For any causal structure $C^\rG$ in which one or more nodes have been split into alternatives by post-selecting on parentless observed nodes, the closure of the set of achievable entropy vectors for the coexisting observed variables is a convex cone.
	\end{cor}

	\section{Quantum causal structures}\label{sec:quantum}
	For quantum causal structures there is an entropy vector method in which the components are the unconditional von Neumann entropies, $H$, of all coexisting sets~\cite{Chaves2015}. Conditional entropies are not explicitly included, but, because $H(A|B)=H(AB)-H(B)$, relations involving conditional entropy can still be encoded. It is natural to ask whether the technique introduced earlier in this paper could yield different results to the existing entropic approach to quantum causal structures. In this section, we consider this question.
	
	Following the approach outlined earlier in the paper we consider including all entropy inequalities from~\cite{Chaves2015} (see Appendix~\ref{app:quantum} for a full description of the method employed in~\cite{Chaves2015}). These are automatically part of our approach, since the unconditional measurement entropy---for which we include all inequalities valid in the theory at hand---coincides with the von Neumann entropy in the quantum case. In addition, we also take into account inequalities for conditional measurement entropies, which are always positive and differ from the conditional von Neumann entropy. Thus, our approach could lead to more restrictive inequalities than the previous quantum one. With the following proposition we show that the previous method for quantum causal structures~\cite{Chaves2015} can be refined in a way that makes the additional variables corresponding to conditional measurement entropies superfluous.
	
	\begin{prop}\label{lem:quantumlemma} Consider a causal structure
		$C^\rQ$ and suppose that, in addition to any causal constraints, we use positivity of unconditional entropies, strong subadditivity and additionally impose positivity of conditional entropies for all combinations of variables that occur in a co-existing set.  If we then eliminate all variables corresponding to unobserved systems, the resulting entropic inequalities for the observed variables are all valid.
	\end{prop}
	
	As a result, all valid inequalities for conditional measurement entropy can be imposed for the conditional von Neumann entropy instead, and can hence be encoded as linear constraints on the (unconditional) von Neumann entropy. In other words, although the conditional von Neumann entropy can be negative for some quantum states, by constraining it to be positive and eliminating unobserved variables, we obtain valid entropy inequalities for the observed variables in $C^\rQ$. This was not used in previous entropic analyses of quantum causal structures~\cite{Chaves2015}.
	
	Previous quantum methods~\cite{Chaves2015,review} instead analysed quantum causal structures by considering the von Neumann entropy of coexisting sets and imposing positivity of the entropy, strong-subadditivity as well as the weak monotonicity constraints that for any state $\rho_{XYZ}$, $H(X|Y)+H(X|Z) \geq 0$~\cite{Pippenger}.  Weak monotonicity constraints are not needed in the statement of Proposition~\ref{lem:quantumlemma} because they are implied by the positivity of conditional entropies. 
	See Appendix~\ref{app:quantumlemma} for the full proof of Proposition~\ref{lem:quantumlemma} and for a complete account of previous quantum methods.
	
	This also gives an important insight into the entropy vector method: the difference between the inequalities that result from using the entropy vector method in the classical~\cite{Chaves2012,Fritz2013} and quantum~\cite{Chaves2015} cases is entirely due to the fact that in the quantum case not all variables coexist and does not arise from the different properties of the Shannon and von Neumann entropy (see also Appendix~\ref{app:quantumlemma} for further discussion).
	
	For most causal structures of interest we can prove that our refined entropy vector method does not allow us to find any tighter entropy inequalities than that of~\cite{Chaves2015} (see Lemma~\ref{lem:different_quant} in Appendix~\ref{app:quantum_entropy_vectors}). Nonetheless, the possibility of using positivity of conditional entropy instead of weak monotonicity simplifies the quantum method even in these cases.

	\section{Applications: Entropic characterisation of causal structures with GPT resources}
	In this section we illustrate the techniques introduced above through a series of examples that are chosen to show the different types of result that can occur.  We summarize our findings here, before presenting these cases in detail.
	
	First, in Section~\ref{sec:applications1}, we consider the case without post-selection and give three types of example:
	\begin{itemize}
		\item A causal structure in which the actual entropic cones for classical, quantum and GPTs are provably the same is in given in Section~\ref{sec:inst}.
		\item A causal structure in which our methods lead to the same outer approximations to the entropy cones in the classical, quantum and GPTs cases is given in Section~\ref{sec:str5}.
		\item Two causal structures for which our methods yield the same outer approximations of the entropy cones in the classical and quantum cases, but a different outer approximation is obtained for GPTs is given in Sections~\ref{sec:str5adapta} and~\ref{sec:str5adaptb}. In these cases we are not aware of any GPT correlations that violate the classical/quantum inequalities.
	\end{itemize}
	Note that we already know of a case (the triangle causal structure) where we have different outer approximations of the entropy cones in the classical and quantum cases~\cite{nonshan}.  We are currently unaware of a case without post-selection in which the outer approximations of the entropy cones between two theories are different and in which we know of correlations in one theory that violate those in another. In other words, it remains possible that the gaps we find in the outer approximations are a symptom of the method used, rather than features of the actual entropy cones.
	
	Then, in Section~\ref{sec:applications2} we move to the use of post-selection and show that:
	\begin{itemize}
		\item For the bilocality causal structure there are different entropy cones for classical, quantum and locally tomographic GPTs (Section~\ref{sec:biloc}). We also identify distributions that certify that the true entropy cones are indeed different for all three cases.
		\item We can use our methods to obtain entropic inequalities for the Information Causality scenario in arbitrary GPTs (Section~\ref{sec:IC}).
	\end{itemize}

	\subsection{Analysis of causal structures without post-selection} \label{sec:applications1}
	In this section we illustrate that our method for GPTs can recover the classically valid entropy inequalities for some causal structures (the first two examples), and that for other causal structures we recover different inequalities (the last two examples). 
	
	\subsubsection{Instrumental Scenario}\label{sec:inst}
	In the classical literature, one well-studied causal structure is the \emph{instrumental causal structure}~\cite{Pearl1995} of Figure~\ref{fig:instrumental}. The quantum version of this causal structure has also been studied~\cite{Chaves2018,VanHimbeeck2019}.
	\begin{figure}
		\centering
		\resizebox{0.5 \columnwidth}{!}{%
			\begin{tikzpicture}[scale=0.8]
			\node[draw=black,circle,scale=0.75]  (Z1) at (-4.5,-2) {$Y$};
			\node[draw=black,circle,scale=0.75]  (Y1) at (-7,-2) {$Z$};
			\node (Z2) at (-4.75,-1.1) {$A_Y$};
			\node (Y2) at (-6.75,-1.1) {$A_Z$};
			\node[draw=black,circle,scale=0.75]  (X1) at (-9.5,-2) {$X$};
			\node (A1) at (-5.75,-0.7) {$A$};
			\draw [->,>=stealth] (A1)--(Y1);
			\draw [->,>=stealth] (A1)--(Z1);
			\draw [->,>=stealth] (Y1)--(Z1);
			\draw [->,>=stealth] (X1)--(Y1);
			\end{tikzpicture}
		}%
		\caption[.] {The instrumental causal structure. The nodes labelled $X$, $Y$ and $Z$ correspond to observations, modelled as random variables. The node $A$ labels a resource-system with subsystems $A_Y$ and $A_Z$ associated to its outgoing edges.}
		\label{fig:instrumental}
	\end{figure}
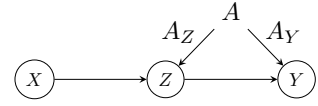
	
	In this case, the coexisting sets are all subsets of $\left\{A_Y, A_Z, X \right\}$, $\left\{A_Y, X , Z \right\}$ and $\left\{X, Y, Z \right\}$.  The second set implies that the entropy vector includes components corresponding to {\footnotesize $\hpl{A_Y}$, $\hpl{X}$, $\hpl{Z}$, $\hpl{A_YX}$, $\hpl{A_YZ}$, $\hpl{XZ}$, $\hpl{A_YXZ}$, $\hpl{A_Y|X}$, $\hpl{A_Y|Z}$, $\hpl{A_Y|XZ}$, $\hpl{A_YX|Z}$, $\hpl{A_YZ|X}$, $\hpl{X|A_Y}$, $\hpl{X|A_YZ}$, $\hpl{XZ|A_Y}$, $\hpl{Z|A_Y}$, $\hpl{Z|A_YX}$}, and similarly for the other two, leading to a vector with 35 components.
	
	For these we impose all entropy inequalities of Properties~1 to 8 as well as the independencies of the subsystems of $A$ and $X$, such as $\hpl{A_YA_Z|X}=\hpl{A_YA_Z}$ or $\hpl{X|A_YA_Z}=\hpl{X}$.
	
	Eliminating all other variables in order to obtain inequalities that only involve the components $(H(X)$, $H(Y)$, $H(Z)$, $H(XY)$, $H(XZ)$, $H(YZ)$, $H(XYZ) )$, we obtain the \emph{Shannon inequalities} (positivity of entropy and conditional entropy, and strong subadditivity) for three variables and
	\begin{equation} \label{eq:inst}
	I(X:YZ) \leq H(Z)\,,
	\end{equation}
	which form a polyhedral cone $\Gamma$. Valid entropy vectors for distributions compatible with the instrumental scenario for a system $A$ of some locally tomographic GPT are necessarily within $\Gamma$.  For this causal structure, it is known that being in $\Gamma$ is necessary and sufficient for being in the closure of the set of valid entropy vectors when $A$ is classical or quantum~\cite{Chaves2014, Henson2014, nonshan}.  Since classical systems are a special case of systems in a GPT, it follows that membership of $\Gamma$ is also sufficient for locally tomographic GPTs.  We have hence found all valid entropy inequalities in this scenario for such theories.
	
	According to Proposition~\ref{prop:generalGPT},~\eqref{eq:inst} holds in any (not necessarily locally tomographic) GPT\footnote{It is already clear that it holds in any GPT satisfying the premise of Property~\ref{it:b1}. Note that Proposition~\ref{prop:generalGPT} implies that the set of correlations in $C^\rB$ is at least as large as those in $C^\rG$, so any valid restriction that holds for correlations in $C^\rB$ holds also for correlations in $C^\rG$.}. Thus, in the instrumental causal structure, $\Gamma$ completely characterises the set of achievable entropy vectors independently of whether system $A$ is classical, quantum or any GPT system.
	
	\subsubsection{Causal structure of Figure~\ref{fig:str5}}\label{sec:str5}
	Applied to the causal structure of Figure~\ref{fig:str5}, our method
	leads to the following entropy inequalities for the observed variables
	when $A$ and $B$ are taken to be box-world systems: 
	\begin{align}
	I(C:E|D)&=0 \\
	I(C:DEF)&\leq H(D)\\
	I(C:EF)&\leq H(E)
	\end{align} 
	and the Shannon inequalities.
	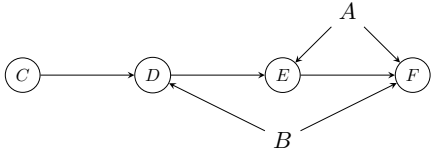
\begin{figure}
		\centering
		\resizebox{0.7 \columnwidth}{!}{%
			\begin{tikzpicture}[scale=0.85]
			\node[draw=black,circle,scale=0.75]  (F) at (-2,-2) {$F$};
			\node[draw=black,circle,scale=0.75]  (E) at (-4.5,-2) {$E$};
			\node[draw=black,circle,scale=0.75]  (D) at (-7,-2) {$D$};
			\node[draw=black,circle,scale=0.75]  (C) at (-9.5,-2) {$C$};
			\node (A) at (-3.25,-0.75) {$A$};
			\node (B) at (-4.5,-3.25) {$B$};
			\draw [->,>=stealth] (C)--(D);
			\draw [->,>=stealth] (D)--(E);
			\draw [->,>=stealth] (E)--(F);
			\draw [->,>=stealth] (A)--(E);
			\draw [->,>=stealth] (A)--(F);
			\draw [->,>=stealth] (B)--(D);
			\draw [->,>=stealth] (B)--(F);
			\end{tikzpicture}
		}%
		\caption[.] {Causal Structure where our outer approximations for classical, quantum and box-world resources coincide.}
		\label{fig:str5}
	\end{figure}
	In~\cite{Henson2014} the same inequalities were derived for classical $A$ and $B$. It follows from Proposition~\ref{lem:quantumlemma} (see below), that these constraints also hold in the quantum case.  Thus, for all three theories we obtain the same outer approximation of the respective entropy cones\footnote{This may change if we were to take non-Shannon inequalities into account prior to Fourier-Motzkin elimination.}. Violation of any of these inequalities excludes this causal structure as a possible explanation of the observed correlations, irrespective of the nature of the pre-shared resources.
	
	In this example, these outer approximations are not tight: there are further valid entropy inequalities for the classical systems $C$, $D$, $E$ and $F$|so-called non-Shannon inequalities|that lead to tighter approximations, e.g.\ the following inequality (from~\cite{Zhang1997}):
	\begin{align}
	{I(D: E|F)} + {I(D : E|C)} + {I(F : C)} - {I(D : E)} \nonumber \\ 
	+{I(D : F|E)} + {I(E : F|D)} + {I(D :  E|F)} \geq 0.
	\end{align}

	\subsubsection{Causal structure of Figure~\ref{fig:str5adapt}(a)}\label{sec:str5adapta}
	With resources $A$ and $B$ that are allowed in box-world we obtain the
	Shannon inequalities and
	\begin{align}
	I(C:EF) &\leq H(E), \\
	I(C:DEF) &\leq H(D).
	\end{align}
	Classical and quantum resources $A$ and $B$ lead to slightly tighter inequalities, namely the Shannon inequalities and 
	\begin{align}
	I(C:EF) &\leq H(E) \\
	I(C:DEF)+I(D:F|E) &\leq H(D).
	\end{align}
	The question of whether or not there exist box-world distributions that
	violate one of these inequalities remains open\footnote{In general, the methods used here lead to outer approximations to the entropy cone of a causal structure. In some cases we can show these to be tight, see e.g.~\cite{linelike, thesis}, but not always. It could thus be the case that while there is a gap between the outer approximations of these cones, there is still no gap between the true cones.}.
	
	\subsubsection{Causal structure of Figure~\ref{fig:str5adapt}(b)}\label{sec:str5adaptb}
	With resources $A$ and $B$ that are allowed according to the theory of box-world we obtain the Shannon inequalities and 
	\begin{align}
	I(C:D) &= 0, \label{eq:18}\\
	H(F|CE) &\leq H(CF|DE),\label{eq:19}\\
	I(D:CEF) &\leq H(E). 
	\end{align}
	Classical and quantum resources lead to slightly tighter inequalities, namely the Shannon inequalities, Equation~\eqref{eq:18}, Inequality~\eqref{eq:19} and 
	\begin{align}
	I(D:CEF) &\leq H(E|C). 
	\end{align}
	Note that the classical case was treated in~\cite{Henson2014}.
	\bigskip

	In general, such comparisons are interesting for analysing the nature of causation
	in different theories.  In a sense the box-world inequalities can be
	thought of as minimal requirements for a theory with a reasonable
	notion of causation. Developing a systematic understanding of this may
	hint at ways to find a physical principle that singles out quantum
	correlations in general scenarios.
	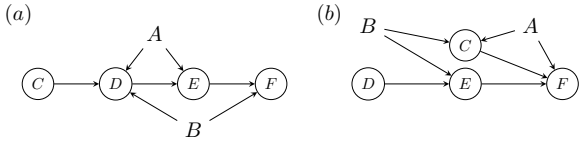
\begin{figure}
		\centering
		\resizebox{0.95 \columnwidth}{!}{%
			\begin{tikzpicture}[scale=0.7]
			\node (N) at (-8.5,-0.2) {$(a)$};
			\node[draw=black,circle,scale=0.75]  (F) at (-2,-2) {$F$};
			\node[draw=black,circle,scale=0.75]  (E) at (-4,-2) {$E$};
			\node[draw=black,circle,scale=0.75]  (D) at (-6,-2) {$D$};
			\node[draw=black,circle,scale=0.75]  (C) at (-8,-2) {$C$};
			\node (A) at (-5,-0.7) {$A$};
			\node (B) at (-4,-3.2) {$B$};
			\draw [->,>=stealth] (C)--(D);
			\draw [->,>=stealth] (D)--(E);
			\draw [->,>=stealth] (E)--(F);
			\draw [->,>=stealth] (A)--(E);
			\draw [->,>=stealth] (A)--(D);
			\draw [->,>=stealth] (B)--(D);
			\draw [->,>=stealth] (B)--(F);
			
			\node (N2) at (-0.5,-0.2) {$(b)$};
			\node[draw=black,circle,scale=0.75]  (NF) at (5.5,-2) {$F$};
			\node[draw=black,circle,scale=0.75]  (NE) at (3,-2) {$E$};
			\node[draw=black,circle,scale=0.75]  (ND) at (0.5,-2) {$D$};
			\node[draw=black,circle,scale=0.75]  (NC) at (3,-1) {$C$};
			\node (NA) at (4.675,-0.5) {$A$};
			\node (NB) at (0.5,-0.5) {$B$};
			\draw [->,>=stealth] (NC)--(NF);
			\draw [->,>=stealth] (ND)--(NE);
			\draw [->,>=stealth] (NE)--(NF);
			\draw [->,>=stealth] (NA)--(NC);
			\draw [->,>=stealth] (NA)--(NF);
			\draw [->,>=stealth] (NB)--(NC);
			\draw [->,>=stealth] (NB)--(NE);
			\end{tikzpicture}
		}%
		\caption[.] {Causal structures where there are different entropic bounds depending on the nature of the unobserved $A$ and $B$.}
		\label{fig:str5adapt}
	\end{figure}
	
	\subsection{Analysis of causal structures with post-selection}\label{sec:applications2}
	In this section we apply our method to post-selected causal structure and show how this allows us to distinguish the correlations obtained in different GPTs. We give two examples for this.\footnote{These constraints can be useful even in the classical case, where there is already the inflation technique for deriving incompatibility constraints~\cite{Wolfe2016, Navascues2017}, if the cardinality of the variables is high, for instance.}
	
	\subsubsection{Bilocality}\label{sec:biloc}
	The bilocal causal structure, first analysed in~\cite{Branciard2010, Branciard2012}, characterises the situation we encounter in scenarios where we rely on entanglement swapping~\cite{Bennett1993}, see Figure~\ref{fig:bilocality}(a). Technologically, this scenario is encountered, for instance, in quantum repeaters~\cite{Briegel1998} or event-ready detection schemes~\cite{Zukowski1993}. 
	
	The entropy vector method provides us with a convenient means to compare the observable correlations when the sources $L_1$ and $L_2$ come from different theories. Applying the entropy vector method to the bilocal causal structure with sources from a locally tomographic GPT, apart from the Shannon inequalities we find only
	\begin{equation}
	H(X_0Z_0)=H(X_0)+H(Z_0)
	\end{equation}
	up to symmetry (exchanging $X_0$ and $X_1$ as well as $Z_0$ and $Z_1$).
	
	Applying the entropy vector method to the bilocal causal structure in the quantum case we find the Shannon inequalities as well as
	\begin{align}
	H(X_0Z_0) &= H(X_0)+H(Z_0) \label{eq:biloc1}\\
	I(X_0Y_0:Z_0) &\leq H(Y_0|X_1) \label{eq:biloc2}\\
	I(X_1:Z_1|Y_0) &\leq H(Y_0|X_0)+H(Y_0|Z_0)-H(Y_0), \label{eq:biloc}
	\end{align}
	up to symmetry. [There are $4$ instances of~\eqref{eq:biloc1} (obtained by exchanging $X_0$ and $X_1$ as well as $Z_0$ and $Z_1$), $16$ of~\eqref{eq:biloc2} (obtained by exchanging $X_0$ and $X_1$, $Y_0$ and $Y_1$, or $Z_0$ and $Z_1$ and by exchanging the roles of $X$ and $Z$) and $8$ of~\eqref{eq:biloc} (obtained by exchanging $X_0$ and $X_1$, $Y_0$ and $Y_1$, or $Z_0$ and $Z_1$).]
	
	The equality constraints are found in both the quantum and GPT case, while the quantum description is tighter, and in this case the gap is provably significant. An example of a box-world distribution that violates~\eqref{eq:biloc} is obtained by taking the systems $L_1$ and $L_2$ to be PR-boxes, $A \in \left\{0,1 \right\}$ is the uniform input and $X\in \left\{0,1 \right\}$ the output on the left, and analogously for $C$ and $Z$ on the right. $B\in \left\{0,1 \right\}$ is with probability $\frac{1}{2}$ input into the first and with probability $\frac{1}{2}$ input into the second box and the respective output serves as an input for the other, where $Y$ is equal to the outputs of the two boxes. This distribution also violates the classical inequalities below and is, to the best of our knowledge, the first manifestation of such a violation with a GPT correlation that is proven to be achievable in the bilocal causal structure. [Note that when the violation of the classical inequalities of~\cite{Chaves2012} was discussed,  
	a tripartite box was considered directly (and without proof that it can be generated from a GPT in the bilocal causal structure)\footnote{Note also that we have a slight difference from~\cite{Chaves2012}: they did not consider $B$.}.]
	
	For classical sources $L_1$ and $L_2$, we obtain a convex cone constrained by Shannon inequalities and $53$ other independent classes of linear inequalities. We list one representative of each of these $53$ classes in Appendix~\ref{app:biloc}. Note that the (in)equalities obtained in the quantum case (cf.~\eqref{eq:biloc1}--\eqref{eq:biloc}) are present. A quantum distribution that leads to an entropy vector outside this cone and which is thus not achievable with classical resources in this causal structure is, for instance, the following: Suppose $L_1$ and $L_2$ share a singlet state each and at the nodes the output is made according to projective measurements $\Pi_{\theta}$ in the basis $\left\{ \cos \theta \ket{0} + \sin \theta \ket{1}, \sin \theta \ket{0} - \cos \theta \ket{1} \right\}$. The angles are chosen as follows: for $X_0$, $\theta=x$, for $X_1$, $\theta=3x$, for $Z_0$, $\theta=0$, for $Z_1$, $\theta=2x$ and for $Y_0$ we consider $\theta=0$ on the state shared by $L_1$ and $\theta=(2y_0+1) x$ on the other, for $Y_1$ we consider $\theta=2 x$ on the state shared by $L_1$ and $\theta=(2y_0+1) x$ on the other system, where the output is made up of the outputs of both measurements in both cases and where $x = 0.1$ and $y_0$ denotes the outcome of the first measurement. This distribution violates several of our inequalities, for instance one inequality from class~31 from Appendix~\ref{app:biloc}). Note also that for the entropic inequalities derived in~\cite{Chaves2012} for the bilocal causal structure with classical sources, no quantum violations are known, so, as far as we know, the characterisation reported here is the first entropic one that is provably able to resolve this gap.

	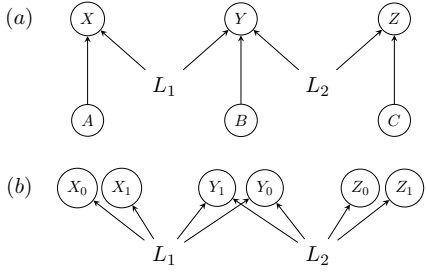
\begin{figure}
		\begin{center}
			\resizebox{0.7\columnwidth}{!}{%
				\begin{tikzpicture}[scale=0.6]%
				\node  (0) at (-17,-1) {$(a)$};
				\node (00) at (-17,-6) {$(b)$};
				\node [draw=black,circle,scale=0.75] (X1a) at (-15,-1) {$X$};
				\node [draw=black,circle,scale=0.75] (X1b) at (-15,-4) {$A$};
				\node [draw=black,circle,scale=0.75] (Y1a) at (-10.5,-1) {$Y$};
				\node [draw=black,circle,scale=0.75] (Y1b) at (-10.5,-4) {$B$};
				\node [draw=black,circle,scale=0.75] (Z1a) at (-6,-1) {$Z$};
				\node [draw=black,circle,scale=0.75] (Z1b) at (-6,-4) {$C$};
				\node (Aa) at (-12.75,-3) {$L_1$};
				\node (Bb) at (-8.25,-3) {$L_2$};
				
				\draw [->,>=stealth] (Aa)--(Y1a);
				\draw [->,>=stealth] (Aa)--(X1a);
				\draw [->,>=stealth] (Bb)--(Y1a);
				\draw [->,>=stealth] (Bb)--(Z1a);
				\draw [->,>=stealth] (Y1b)--(Y1a);
				\draw [->,>=stealth] (Z1b)--(Z1a);
				\draw [->,>=stealth] (X1b)--(X1a);
				
				\node [draw=black,circle,scale=0.75] (X1) at (-14,-6) {$X_1$};
				\node [draw=black,circle,scale=0.75] (X0) at (-15.3,-6) {$X_0$};
				\node [draw=black,circle,scale=0.75] (Y0) at (-9.8,-6) {$Y_0$};
				\node [draw=black,circle,scale=0.75] (Y1) at (-11.2,-6) {$Y_1$};
				\node [draw=black,circle,scale=0.75] (Z0) at (-7,-6) {$Z_0$};
				\node [draw=black,circle,scale=0.75] (Z1) at (-5.7,-6) {$Z_1$};
				\node (A) at (-12.75,-8) {$L_1$};
				\node (B) at (-8.25,-8) {$L_2$};
				
				\draw [->,>=stealth] (A)--(Y0);
				\draw [->,>=stealth] (A)--(X0);
				\draw [->,>=stealth] (A)--(Y1);
				\draw [->,>=stealth] (A)--(X1);
				\draw [->,>=stealth] (B)--(Y0);
				\draw [->,>=stealth] (B)--(Z0);
				\draw [->,>=stealth] (B)--(Y1);
				\draw [->,>=stealth] (B)--(Z1);
				\end{tikzpicture}
			}
		\end{center}
		\caption[.] {Bilocal causal structure. (a) The nodes labelled $A$, $B$, $C$, $X$, $Y$ and $Z$ correspond to observations, modelled as random variables. The nodes $A$ and $B$ label a resource-system which in this case we take to be quantum. (b) Post-selected version of the bilocal causal structure, where $X_0$ stands for $X_{|A=0}$ and similarly for $X_1$ as well as the corresponding $Y$ and $Z$.}
		\label{fig:bilocality}
	\end{figure}
	
	\subsubsection{Information Causality}\label{sec:IC}
	Information causality~\cite{Pawlowski2009} is a candidate principle for singling out quantum theory. Roughly speaking the principle is that that sending $n$ bit of classical information from one party to another cannot give the recipient access to more than $n$ bits of previously unknown information regardless of any pre-shared resources the parties may have. The associated causal structure is shown in Figure~\ref{fig:infcaus}, and information causality is obeyed if
	\begin{align}
	I(X_1:Y_{|R=1})+I(X_2:Y_{|R=2})\leq H(Z)\,.
	\end{align}
	This relation is known to hold in both classical and quantum theory, while it is violated in box-world.\footnote{In fact, tighter entropy inequalities have since been shown to hold in the quantum (and classical) case~\cite{Chaves2015}.}
	
	Using the technique of the present paper, we find that the following relations hold for underlying box-world systems
	\begin{align}
	I(X_1X_2:Y_{|R=1}Z) &\leq H(Z) \label{eq:IC1} \\
	I(X_1X_2:Y_{|R=2}Z) &\leq H(Z) \label{eq:IC2} \, .
	\end{align}
	These are valid in any GPT (see Proposition~\ref{prop:generalGPT}) for the entropy vectors ${\bf{H}}$ containing the entropies of all $23$ subsets of the coexisting sets $\{X_1,X_2,Y_{|R=1},Z\}$ and $\{X_1,X_2,Y_{|R=2},Z\}$. Thus, although information causality does not hold in general, some minimal notion of causation remains (beyond no-signalling).
	
	We remark that the information causality scenario in boxworld was also considered in~\cite{Chaves2016}, but in a slightly different way.  There the relations ${I(X_1:Y_{|R=1})} \leq H(Z)$ and ${I(X_2:Y_{|R=2})} \leq H(Z)$ were postulated. We are able to recover these with our approach in the following way.  If, instead of considering all joint entropies of coexisting variables, only the restricted entropy vectors with components ${\bf{H}_R}=(H(X_1)$, $H(X_2)$, $H(Y_{|R=1})$, $H(Y_{|R=2})$, $H(Z)$, $H(X_1 Y_{|R=1})$, $H(X_2Y_{|R=2}))$ are considered, the inequalities ${I(X_1:Y_{|R=1})}\leq H(Z)$ and ${I(X_2:Y_{|R=2})}\leq H(Z)$ emerge. Since all extremal vertices of the entropy cone of vectors ${\bf{H}_R}$ are achievable (as was shown in~\cite{Chaves2016}) and, according to Corollary~\ref{cor:convexity}, the (closure of the) set of achievable entropy vectors ${\bf{H}_R}$ is convex, this is the true entropy cone for this restricted marginal scenario.  However, we don't see a clear motivation for excluding the additional observed entropies.
	
	The two inequalities~\eqref{eq:IC1} and~\eqref{eq:IC2} were already derived in~\cite{Short2010a} for box-world, but here they emerge systematically from our method (and our results imply that they hold for any GPT). They are the only inequalities (other than Shannon inequalities) that follow from our method as it was introduced above. However, further inequalities hold for this scenario, for instance the non-Shannon inequality~\cite{Zhang1997}
	\begin{align*}
	I(X_2:Y_{|R=1}|Z)+I(X_2:Y_{|R=1}|X_1)+I(Z:X_1)\nonumber\\
	-I(X_2\!:\!Y_{|R=1})\!+\!I(X_2\!:\!Z|Y_{|R=1})\!+\!I(Y_{|R=1}\!:\!Z|X_2) \! \\
	+I(X_2:Y_{|R=1}|Z)\geq0\,. 
	\end{align*}
	Because a complete set of non-Shannon inequalities is not known, we do
	not have a complete characterisation of the entropy cone of vectors
	${\bf{H}}$ for this scenario (which may require infinitely many linear
	inequalities).
	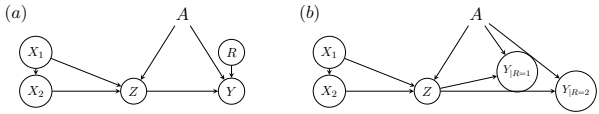
\begin{figure}
		\centering
		\resizebox{0.975 \columnwidth}{!}{%
			\begin{tikzpicture}[scale=0.85]
			\node (a)  at (-17.5,0) {$(a)$};
			\node[draw=black,circle,scale=0.75]  (Y20) at (-12,-2) {$Y$};
			\node[draw=black,circle,scale=0.75]  (R20) at (-12,-1) {$R$};
			\node[draw=black,circle,scale=0.75]  (Z10) at (-14.5,-2) {$Z$};
			\node[draw=black,circle,scale=0.75]  (X10) at (-17,-1) {$X_1$};
			\node[draw=black,circle,scale=0.75]  (X20) at (-17,-2) {$X_2$};
			\node (A10) at (-13.25,0) {$A$};
			\draw [->,>=stealth] (A10)--(Y20);
			\draw [->,>=stealth] (A10)--(Z10);
			\draw [->,>=stealth] (Z10)--(Y20);
			\draw [->,>=stealth] (X10)--(X20);
			\draw [->,>=stealth] (X10)--(Z10);
			\draw [->,>=stealth] (X20)--(Z10);
			\draw [->,>=stealth] (R20)--(Y20);

			\node (b)  at (-10,0) {$(b)$};
			\node[draw=black,circle,scale=0.65]  (Y1) at (-4.7,-1.5) {$Y_{|R=1}$};
			\node[draw=black,circle,scale=0.65]  (Y2) at (-3.2,-2) {$Y_{|R=2}$};
			\node[draw=black,circle,scale=0.75]  (Z1) at (-7,-2) {$Z$};
			\node[draw=black,circle,scale=0.75]  (X1) at (-9.5,-1) {$X_1$};
			\node[draw=black,circle,scale=0.75]  (X2) at (-9.5,-2) {$X_2$};
			\node (A1) at (-5.75,0) {$A$};
			\draw [->,>=stealth] (A1)--(Y1);
			\draw [->,>=stealth] (A1)--(Y2);
			\draw [->,>=stealth] (A1)--(Z1);
			\draw [->,>=stealth] (Z1)--(Y1);
			\draw [->,>=stealth] (Z1)--(Y2);
			\draw [->,>=stealth] (X1)--(X2);
			\draw [->,>=stealth] (X1)--(Z1);
			\draw [->,>=stealth] (X2)--(Z1);
			\end{tikzpicture}
		}%
		\caption[.] {Causal structure of the Information Causality
			scenario. (a) Alice holds two pieces of information $X_1$ and $X_2$
			and is allowed to send a message $Z$ to Bob. Bob then has to make a
			guess of either $X_1$ or $X_2$, depending on the request of a
			referee, $R=1$ or $R=2$. $A$ is a pre-shared resource the parties
			may use. (b) We divide $Y$ into two variables, $Y_{|R=1}$ and
			$Y_{|R=2}$, depending on the question $R$. While for
			classical $A$ Bob can always compute the value of both $Y_{|R=1}$
			and $Y_{|R=2}$, more generally these have to be understood as
			alternatives, of which only one is generated.}
		\label{fig:infcaus}
	\end{figure}

	\section{Limitations and directions to overcome them}
	As is the case for previous entropic methods~\cite{Chaves2012, Fritz2013, Chaves2015}, there are causal structures for which this method does not imply any entropic constraints for the observed variables (except Shannon inequalities), an example being the triangle causal structure~\cite{Steudel2015, Fritz2012, Henson2014, Chaves2015, Wolfe2016, nonshan, Gisin2017}.\footnote{However, the triangle causal structure with inputs can be treated with our method.}  Furthermore, all known strategies that certify incompatibility of entropy vectors relying on GPT resources with classical and quantum scenarios rely on post-selection (cf.\ Figure~\ref{fig:infcaus}). If post-selection is necessary for this, for some causal structures (such as the triangle causal structure) current entropic techniques cannot certify this distinction. This is not a severe limitation, since most experimentally interesting causal structures involve measurement settings, which we can post-select on.
	
	Considering entropy vectors rather than the corresponding joint probability distributions gives a computational advantage and provides constraints that are valid independently of the dimension of the involved resources. However, this advantage comes with restricted precision (see for instance~\cite{linelike}).  In particular, there are distributions between observed variables that are realisable with box-world resources but not with classical or quantum systems but which the method cannot certify as such. How to overcome this remains an open question, although the ideas in~\cite{Chaves2013} may form a useful starting point (notwithstanding the limitations of the ideas of~\cite{Chaves2013} identified in~\cite{VC2019b}).
	
	While we found that our method strictly improves on previous entropic methods~\cite{Chaves2016}, another promising research avenue is to generalise the inflation technique to GPTs. This research has been started in~\cite{Wolfe2016}, where it was pointed out that certain inflations are valid for any GPT and hence some general constraints can be derived with it. The question of how these constraints relate to the ones found with our method is left for future work. One key difference is that the method presented in~\cite{Wolfe2016} does not distinguish what GPT the latent systems are described by (e.g.\ whether they are quantum and box-world systems). Considering the latent variables explicitly allows us to make this distinction and to certify that different sets of correlations are produced within different GPTs (e.g.\ quantum theory and box-world).

	\bigskip
	
	\section*{Acknowledgements}
	We thank Peter J.~Brown for comments on an earlier version of this manuscript. This work was supported by an EPSRC First grant (grant number EP/P016588/1), the EPSRC’s Quantum Communications Hub (grant number M013472/1) and by the Austrian Science fund (FWF) stand-alone project P~30947. The majority of this work was carried out while MW was based at the University of York.%
	

	\onecolumngrid
	
	\appendix

	\section{Further details regarding the inequalities for the measurement entropy in GPTs} \label{app:general}
	
	\label{app:boxworld} \label{app:lemcoarse}
	
	Before getting to the additional properties, we need a few lemmas. The
	first is the concavity of $\hplo$, proven in~\cite{Short2010a,
		Barnum2010}, which follows from the concavity of the Shannon
	entropy.
	
	\begin{lem}\label{lem:concavity}
		For $s_1$, $s_2\in\sta_A$ and $0 \leq p \leq 1$,
		$\hpl{A}_{p s_1+(1-p) s_2} \geq p \hpl{A}_{s_1}+ (1-p) \hpl{A}_{s_2}$.
	\end{lem}
	
	The next lemma says that the infimum in the definition of conditional
	measurement entropy can be restricted to fine-grained measurements.
	
	\begin{lem} \label{lem:coarse}
		Let $s_{AB} \in \sta_{AB}$, then $\hpl{A|B}=\inf_{\{f^y\} \in \me^*_B}\sum_y f^y(s_B)\hpl{A_{|y}}$.
	\end{lem}
	\begin{proof}
		Let $\{f^y\}_{y=1}^n \in \me_B$ and consider the coarse-graining
		that combines the last two effects in $\{f^y\}_{y=1}^n$ to give a
		new measurement $\{g^y\}_{y=1}^{n-1}$ where $g^y=f^y$ for
		$y=1,\ldots,n-2$ and $g^{n-1}=f=f^{n-1}+f^n$.  We have
		\begin{align*}
		\sum_{y=1}^{n-1}&g^y(s_B) \hpl{A_{|y}}=\sum_{y=1}^{n-2}f^y(s_B)\hpl{A_{|y}}+(f^{n-1}(s_B)+f^n(s_B))\hpl{\alpha s_{A|n-1}+(1-\alpha)s_{A|n}}\,,
		\end{align*}
		where $\alpha=f^{n-1}(s_B)/(f^{n-1}(s_B)+f^n(s_B))$.  Using
		concavity, we have
		\begin{align*}
		\sum_{y=1}^{n-1}g^y(s_B) \hpl{A}_{s_{A|y}}&\geq\sum_{y=1}^{n-2}f^y(s_B)\hpl{A}_{s_{A|y}}+f^{n-1}(s_B)\hpl{s_{A|n-1}}+f^n(s_B)\hpl{s_{A|n}}\\&=\sum_{y=1}^nf^y(s_B)\hpl{A}_{s_{A|y}},
		\end{align*}
		i.e., for this coarse-graining the measurement on Bob cannot
		decrease the expected measurement entropy on Alice conditioned on
		the result.  Since all coarse-grainings can be formed by a sequence
		of such combinations, it follows that the infimum on Bob's
		measurements can be restricted to fine-grained measurements.
	\end{proof}
	
	It is also worth noting the following.
	
	\begin{lem}\label{lem:inf_FG_triv}
		Let $\{f^y\}\in\me_B$ be a trivial refinement of $\{e^x\}\in\me_B$.
		Then
		$\sum_yf^y(s_B) \hpl{A_{|y}}=\sum_xe^x(s_B)\hpl{A_{|x}}$.
	\end{lem}
	\begin{proof}
		Consider the case in which one of the effects in $\{e^x\}_{x=1}^n$
		is split into two to form $\{f^y\}_{y=1}^{n+1}$ with $f^y=e^y$ for
		$y=1,\ldots,n-1$ and $f^n+f^{n+1}=e^n$.  We have
		$\hpl{A_{|n+1}}=\hpl{A_{|n}}$ and hence in this case the claim
		follows from $f^{n+1}+f^n=e^n$.  Since any trivial refinement can be
		formed by combining such splittings, the result generalizes to all
		trivial refinements.
	\end{proof}

	The following lemma is in essence a restatement of part of the proof
	of Lemma~2 from~\cite{Cadney2012a}.
	
	\begin{lem}\label{lem:classfirst}
		Let $A$ be classical and $B$ be a system from an arbitrary GPT.  For
		any measurement $M\in\me_{AB}$ there exists an $n$-outcome
		measurement $N_A\in\me_A$ and measurements $N^x_B\in\me_B$ for
		$x=1,2,\ldots,n$ such that $M$ is equivalent to performing $N_A$,
		then performing $N^x_B$ (where $x$ is the result of $N_A$).
	\end{lem}
	
	\begin{proof}
		Let $M=\{e^r\}$ and let $\{f^x\}$ be a standard classical measurement
		on $A$.  We can define a new set of effects on $B$ that act as
		$e^r_x:s_B\mapsto e^r(s_A^x\ot s_B)$, where $\{s_A^x\}$ are a set of
		states on $A$ for which $f^x(s_A^y)=\delta_{x,y}$.  For each $x$,
		the set $\{e^r_x\}_r$ form a measurement on $B$:
		$$\sum_re^r_x(s_B)=\sum_re^r(s_A^x\ot s_B)=1\quad\text{for all }s_B\in\sta_B\,.$$
		
		If we take $N_A=\{f^x\}$ and $N_B^x=\{e_x^r\}_r$, then this is
		equivalent to measuring $M$:
		\begin{align*}
		p(r)&=\sum_xp(x)p(r|x)=\sum_xp(x)e_x^r(s_{B|x})=\sum_xf^x(s_A)e^r(s_A^x\ot
		s_{B|x})=\sum_xe^r(s_A^x\ot
		(f^x\ot I_B)(s_{AB}))\\
		&=\sum_xe^r((s_A^xf^x\ot I_B)(s_{AB}))=e^r(s_{AB})\,,
		\end{align*}
		where we have used that $\sum_x s_A^xf^x$ is the identity
		transformation on the classical system $A$.
	\end{proof}

	We will in particular rely on the following corollary of this lemma.
	
	\begin{cor}\label{cor:classfirst}
		Let $A$ be classical and $B$ be a system from an arbitrary GPT. For
		any fine-grained measurement $M\in\me^*_{AB}$ there exists an
		$n$-outcome fine-grained measurement $N_A\in\me^*_A$, and
		fine-grained measurements $N^x_B\in\me^*_B$ for $x=1,2,\ldots,n$
		such that $M$ is equivalent to performing $N_A$, then performing
		$N^x_B$ (where $x$ is the result of $N_A$).
	\end{cor}
	\begin{proof}
		We have already shown that $N_A$ can be taken to be fine-grained.
		Suppose $N_B^1$ is not fine grained, and consider $\hat{N}_B^1$ in
		which $e_1^t$ is split into other effects $e_1^{t_1}$ and
		$e_1^{t_2}$ satisfying $e_1^{t_1}+e_1^{t_2}=e_1^t$ in a non-trivial
		way, i.e., with $e_1^{t_1}$ and $e_1^{t_2}$ not proportional to one
		another or to any other effect $e_1^j$ with $j\neq t$.  The
		measurement that involves measuring $N_A$ and then $\hat{N}_B^1$ if
		$x=1$ and otherwise $N_B^x$ depending on the result is a non-trivial
		fine-graining of $M$.  Thus, if $M$ is fine-grained, so are $N^x_B$
		for all $x$.
	\end{proof}
	
	\begin{proof}[Proof of Property~\ref{it:b1}]
		For any measurements $\{e^x\}\in\me^*_A$ and $\{f^y\}\in\me^*_B$ and
		any state $s_{AB}\in\sta_{AB}$ we have
		\begin{align*}
		-\sum_xe^x(s_A)\log(e^x(s_A))-\sum_yf^y(s_B)\log(f^y(s_B))&\geq-\sum_{xy}(e^x\ot
		f^y)(s_{AB})\log((e^x\ot f^y)(s_{AB}))\\
		&\geq\hpl{AB}\,,
		\end{align*}
		where the first inequality uses subadditivity of the Shannon entropy
		and the last line follows because $\{e^x\ot f^y\}_{x,y}$ is a
		measurement on $AB$ and is hence at least as large as the infimum over
		joint measurements in $\hpl{AB}$.  Then the fact that $\hpl{A}$ and
		$\hpl{B}$ are infima over fine-grained measurements means that
		$\{e^x\}$ and $\{f^y\}$ can be chosen such that
		$-\sum_xe^x(s_A)\log(e^x(s_A))-\sum_yf^y(s_B)\log(f^y(s_B))$ is
		arbitrarily close to $\hpl{A}+\hpl{B}$ implies result.
	\end{proof}
	
	\begin{proof}[Proof of Property~\ref{it:b2}]
		Using Corollary~\ref{cor:classfirst}, the measurement on $BC$ in
		$\hpl{A|BC}$ can be decomposed into a standard measurement on $B$,
		yielding $y$, followed by a fine grained measurement on $C$
		depending on the value of $y$ obtained, i.e.,
		$$\hpl{A|BC}=\sum_ye^y(s_B)\inf_{\{h^z_y\}_z\in\me_C^*}\sum_zh^z_y(s_{C|y})H(A_{|yz})\,.$$
		For some fixed set $\{h^z_y\}_z\in\me_C^*$, the right hand side
		(without the $\inf$) takes the form
		\begin{align*}
		\sum_{yz}p(y,z)H(A_{|yz})&=H(A|YZ)=H(AY|Z)-H(Y|Z)\geq
		H(AB|Z)-H(Y)\\&=H(AB|Z)-\hpl{B}\geq\hpl{AB|C}-\hpl{B}\,,
		\end{align*}
		where we have used that there is
		a standard measurement achieving the infimum for classical systems in
		$\hpl{B}$ and the subadditivity of the Shannon entropy.  The
		result follows because we can choose $\{h^z_y\}_z\in\me_C^*$ such that
		the left hand side is arbitrarily close to $\hpl{A|BC}$.
	\end{proof}
	
	\begin{proof}[Proof of Property~\ref{it:b4}]
		We start from the definition of $\hpl{AB|C}$, and use
		Corollary~\ref{cor:classfirst} to give
		\begin{align*}
		\hpl{AB|C}\!&=\! \! \! \inf_{\{f^z\}\in\me_C,\{e^x_z\}_x\in\me_{AB}^*} \! \! \! \! \! -\sum_zf^z(s_C)\sum_xe_z^x(s_{AB|z})\log(e_z^x(s_{AB|z}))\\
		&= \! \! \! \inf_{\{f^z\}\in\me_C,\{g^x_{yz}\}_x\in\me_A^*} \! \! \! \! \! -\sum_zf^z(s_C)\sum_{xy}h^y(s_{B|z})g_{yz}^x(s_{A|yz})\log(h^y(s_{B|z})g_{yz}^x(s_{A|yz}))\\
		&= \! \! \! \inf_{\{f^z\}\in\me_C,\{g^x_{yz}\}_x\in\me_A^*} \! \! \! \! \! - \! \sum_zf^z(s_C)  \! \! \left[\!\sum_y \! h^y(s_{B|z})\log(h^y(s_{B|z}))  \!+ \! \!\sum_{xy} \! h^y(s_{B|z})g_{yz}^x(s_{A|yz})\log(g_{yz}^x(s_{A|yz}))\!\right]\\
		&\geq\inf_{\{f^z\}\in\me_C} -\sum_zf^z(s_C)\sum_yh^y(s_{B|z})\log(h^y(s_{B|z}))\quad+\\&\qquad\qquad\qquad\qquad\qquad \inf_{\{f^z\}\in\me_C,\{g^x_{yz}\}_x\in\me_A^*}-\sum_zf^z(s_C)\sum_{xy}h^y(s_{B|z})g_{yz}^x(s_{A|yz})\log(g_{yz}^x(s_{A|yz}))\\
		&\geq\hpl{B|C}+\hpl{A|BC}\,.
		\end{align*}
		In the last inequality we use that a measurement on $C$ followed by a
		fine-grained measurement on $A$ is a joint measurement on $s_{AC}$, so
		the infimum over all joint measurements cannot be larger than this term.
		
		If $C$ is classical then we can drop $\inf_{\{f^z\}\in\me_C}$ and
		take $C$ to always be measured with a standard classical
		measurement. This gives equality in both inequalities in the above
		proof.
	\end{proof}

	\section{Proof of Theorem~\ref{thm:convexity} and Corollary~\ref{cor:convexity}}\label{app:generalGPT}
	\label{app:convexity}
	The proof relies on the following Lemmas.
	
	\begin{lem} \label{lem:1}
		For any causal structure $C^\rG$ where all observed variables share
		one observed ancestor the topological closure of the set of achievable
		entropy vectors of the observed variables is convex.
	\end{lem}
	
	\begin{proof}
		We first prove convexity, and then show that the set form a cone.
		
		Let $C^\rG$ have $n$ observed variables and $m$ unobserved ones. Let
		$\bf{H_1}$ and $\bf{H_2}$ be two achievable entropy vectors for the
		$n$ observed variables in $C^\rG$. In the following, we show that
		for any $0 \leq p \leq 1$, there is a sequence of entropy vectors
		$\bf{H'_k}$ within $C^\rG$, such that
		$\lim_{k \rightarrow
			\infty}{\bf{H'_k}}=p{\bf{H_1}}+(1-p){\bf{H_2}}$.
		
		For $i=1,2$, suppose that ${\bf H_i}$ is generated by using states
		$\{Y_j^i\}_{j=1}^m$ for the $m$ unobserved nodes and that the observed
		random variables are $\{X_j^i\}_{j=1}^n$.  The strategy for achieving
		the convex combination is as follows.  The common observed variable is
		taken to be $X_1'(k)=(A,Z)$ where
		$$(A,Z)=\begin{cases}(0,0)&\text{ with probability }\frac{k-1}{k}\\(1,(X_1^1)^k)&\text{ with probability }\frac{p}{k}\\(2,(X_1^2)^k)&\text{ with probability }\frac{1-p}{k}\end{cases}$$
		and where $X^k$ denotes $k$ i.i.d.\ copies of a random variable $X$.
		Each of the unobserved nodes is prepared in state
		$$Y'_j=(Y_j^1)^{\ot k}\ot(Y_j^2)^{\ot k}\,.$$
		
		Each of the other observed nodes then behaves as follows.  The
		children of $X_1$ have access to $A$.  If $A=0$ they output $X_j'(k)=(0,0)$.
		If $A=1$ they perform the operation that would have led to ${\bf H_1}$
		$k$ times independently, acting on the first $k$ subsystems of any GPT
		resources they have access to.  They then output $X_j'(k)=(1,(X_j^1)^k)$.  If
		$A=2$, the procedure is the same except that the operation that would
		have led to ${\bf H_2}$ is repeated $k$ times by acting on the second
		$k$ subsystems of any GPT resources and the output is
		$X_j'(k)=(2,(X_j^2)^k)$.  Note that the first part of the argument is equal to
		$A$, so, in this way, the value of $A$ is transferred to all
		descendants.  An analogous strategy is then used for subsequent
		generations.
		
		For any subset $S$ of the observed random variables
		$\{X_j'(k)\}_{j=1}^n$ we  have
		$$H_k'(S)=H_k'(A)+pH^1(S)+(1-p)H^2(S)\,,$$
		where $H_k'$ refers to the entropy in the new strategy and $H^1$ and
		$H^2$ refer to the entropies in the original strategies (i.e.,
		according to ${\bf H_1}$ or ${\bf H_2}$).  Noting that
		$$H_k'(A)=-\frac{k-1}{k}\log\frac{k-1}{k}-\frac{p}{k}\log\frac{p}{k}-\frac{1-p}{k}\log\frac{1-p}{k}$$
		tends to $0$ as $k$ tends to $\infty$, we have $\lim_{k \rightarrow
			\infty}{\bf{H'_k}}=p{\bf{H_1}}+(1-p){\bf{H_2}}$.
		
		If $\bf{H_1}$ and $\bf{H_2}$ are themselves only achievable as limits
		of entropy vectors the above argument can be followed for each vector
		in the corresponding sequences tending to $\bf{H_1}$ and $\bf{H_2}$
		respectively and thus also holds for $\bf{H_1}$ and $\bf{H_2}$.  This
		shows that the closure of the set of entropy vectors is convex.
	\end{proof}
	
	The next lemma extends this beyond the case where there is a common
	observed ancestor.
	
	\begin{lem} \label{lem:2}
		For any causal structure $C^\rG$ the topological closure of the set of
		achievable entropy vectors of the observed variables is convex.
	\end{lem}
	
	\begin{proof}
		If all observed variables in $C^\rG$ have a common observed
		ancestor, the statement follows by Lemma~\ref{lem:1}.  Otherwise,
		there are $1 < l \leq n$ observed nodes without any observed
		ancestors, which we label $X_1, \ldots, X_l$ (all other observed
		nodes ($X_{l+1},\ldots,X_n$) are descendants of at least one of
		these nodes). We construct a larger causal structure $C'$ by
		introducing an observed parent node $A_i$ for each $X_i$ with
		$i=1,\ldots,l$, where $A_i$ has no direct link to any variable
		except for $X_i$. Note that a distribution over the observed
		variables $X_1, \ldots, X_n$ is compatible with $C^\rG$ if and only
		if it is the marginal of a distribution over
		$X_1, \ldots, X_n, A_1, \ldots, A_l$ that is compatible with $C'$.
		
		Now let $C''$ be another causal structure that is constructed from
		$C'$ by adding a directed link from $A_1$ to all other $A_i$ with
		$2\leq i \leq l$. A distribution over
		$X_1, \ldots, X_n, A_1, \ldots, A_l$ is compatible with $C'$ if and
		only if it is compatible with $C''$ and, at the same time, obeys
		$I(A_1:A_i)=0$ for all $2\leq i\leq l$.
		
		The ``if'' condition follows
		because any distribution in $C'$ obeys $I(A_1:A_i)=0$ for all
		$2\leq i\leq l$ and it can be realised in $C''$ without using the
		additional causal links $A_1\rightarrow A_i$.
		
		For the ``only if'', we use that $I(A_1:A_i)=0$ holds if and only if
		$p(a_i a_1)=p(a_i)p(a_1)$\footnote{This can for instance be seen by
			writing the mutual information in terms of the relative entropy
			and using Klein's inequality~\cite{Klein1931}.}, so that any
		distribution in $C''$ obeying $I(A_1:A_i)=0$ for all
		$2\leq i\leq l$ can be written as
		\begin{equation}
		p(x_1, \ldots, x_n,a_1, \ldots, a_l)=p(x_{l+1},\ldots,
		x_n|x_1,\ldots,x_l)p(x_1 \ldots x_l|a_1 \ldots a_l)p(a_1) \ldots p(a_l) \, ,
		\end{equation}
		with the right hand side compatible with $C'$.
		
		Hence, a distribution over $X_1, \ldots, X_n$ is compatible with $C^\rG$ if and only if it is the marginal of a distribution over $X_1, \ldots X_n, A_1, \ldots, A_l$ that is compatible with $C''$ and obeys $I(A_1:A_i)=0$ for all $2\leq i\leq l$. 
		
		The closure of the set of entropy vectors of the observed variables
		that are compatible with $C''$ (without any additional constraints) is
		convex by Lemma~\ref{lem:1}. The closure of the set of achievable
		entropy vectors in $C^\rG$ is the closure of the set of achievable
		entropy vectors in of $C''$ restricted by the linear equalities
		$I(A_1:A_i)=0$ for all $2\leq i\leq l$ and projected to the marginals
		involving only $X_1, \ldots, X_n$.  Because these operations preserve
		convexity, the closure of the set of achievable entropy vectors of
		$C^\rG$ is convex. 
	\end{proof}
	
	The main theorem of this section now follows as a corollary.\smallskip

	\begin{proof}[Proof of Theorem~\ref{thm:convexity}]
		Convexity of the set of achievable entropy vectors follows by
		Lemma~\ref{lem:2}. That it is a cone follows because if ${\bf H}$ is
		an achievable entropy vector, then $k{\bf H}$ for $k\in\N$ is
		achievable by taking $k$ independent copies of all systems in the
		strategy achieving ${\bf H}$.  Furthermore, in any causal structure
		$C^\rG$, ${\bf H}={\bf 0}$ is achievable by taking all observed
		variables to be $0$ with probability 1.  Hence, by taking
		convex combinations, if ${\bf H}$ is achievable, so is
		$\lambda{\bf H}$ for any $\lambda\in\R_{\geq 0}$.
	\end{proof}
	
	Corollary~\ref{cor:convexity} then follows in a similar way.

	\begin{proof}[Proof of Corollary~\ref{cor:convexity}]
		Consider first postselecting on one of the parentless variables in $C^\rG$, and suppose that this variable has $k$ possible values. Let $X_1,\ldots,X_n$ be the set of all the observed descendants of the variable that has been postselected on and $Y_1,\ldots,Y_m$ be the set of all other observed nodes.  In other words, for a fixed distribution, $P_{Y_1,\ldots,Y_m}$, of all other observed nodes we consider $k$ different ways to form $X_1,\ldots,X_n$ to give $P_{X_1^1 \cdots X_{n-1}^1 X_n^1Y_1^1,\cdots,Y_m^1}, \ldots, P_{X_1^k \cdots X_{n-1}^k X_n^kY_1^k,\cdots,Y_m^k}$ respectively. For these distributions we define an entropy vector with $k(2^{n+m}-1)$ components by concatenating the entropy vectors of each of them.  Since the marginal distributions obey $P_{Y_1^1\cdots Y_m^1}=\cdots=P_{Y_1^k \cdots Y_m^k}$, $(k-1)(2^m-1)$ components can be removed from the vector.  If ${\bf{H_1}}$ and ${\bf{H_2}}$ are two such achievable entropy vectors, then for any $0 \leq p \leq 1$, $p{\bf{H_1}}+(1-p){\bf{H_2}}$ is also such an achievable entropy vector. This follows by applying the technique used to prove Lemma~\ref{lem:2} separately to the causal structure including only one of the $k$ alternatives (i.e., the causal structure formed from the post-selected causal structure by removing all nodes associated with other alternatives). This strategy leads to the same distribution on $Y_1,\ldots,Y_m$ for each alternative and hence the overall entropy vector has the same entropies for all subsets of $\{Y_1,\ldots,Y_m\}$.
		We can then postselect on further parentless variables in a similar way.
	\end{proof}

	\section{Proof of Proposition~2} \label{app:quantumlemma}
	
	We use $X_1, \ldots, X_n$ for the observed variables and
	$Y_1, \ldots, Y_m$ for the unobserved nodes in $C$.  For each
	unobserved node $Y_i$ we use $Y_i^j$ with $1 \leq j \leq k_i$ for
	the subsystems associated with the $k_i$ outgoing edges, sometimes
	using $Y=\{Y_i^j\}_{i=1,j=1}^{m,\ \ k_i}$ and $X=X_1,\ldots,X_n$ as a
	shorthand. For any unobserved node $Y_i$ with $1 \leq i \leq m$ and
	for any $1 \leq j \leq k_i$, we show how to modify $Y_i^j$ to
	${\tilde{Y}_i^j}$ such that, if ${\tilde{Y}_i^j}$ is shared along
	the $j^{\textrm{th}}$ outgoing edge instead of ${{Y}_i^j}$, the same
	distributions among the observed variables are obtained. This
	construction of ${\tilde{Y}_i^j}$ will make all conditional
	entropies of unobserved systems positive.
	
	If $H({Y}_i^j| Y \cup X \setminus {Y}_i^j) \geq 0$, we set
	$\tilde{Y}_i^j={Y}_i^j$. Otherwise, we let
	$\tilde{Y}_i^j=Y_i^j \otimes \alpha_i^j$, where
	$\alpha_i^j=\sum_a p_i^j(a) \ketbra{a}{a}$
	is a system that is uncorrelated with any other system and obeys
	$H(\alpha_i^j)=H(Y_i^j)$. 
	Then, $H(\tilde{Y}_i^j)=H({Y}_i^j)+H(\alpha_i^j)$ and
	$\tilde{Y}_i^j$ can be used to produce the same observed
	distributions as ${Y}_i^j$, since $\alpha_i^j$ may be ignored when
	processing the unobserved systems $\tilde{Y}_i^j$ to obtain observed
	variables.  Furthermore, due to the independence of the $\alpha_i^j$
	and by weak monotonicity, for any $\tilde{Y}_i^j$ and any set of
	variables, $S$, coexisting with $\tilde{Y}_i^j$,
	\begin{align}
	H(\tilde{Y}_i^j | S)=H({Y}_i^j| S)+H(\alpha_i^j )\geq - H(Y_i^j) + H(\alpha_i^j)=0, \label{eq:onevar}
	\end{align}
	where the last equality follows by construction. (Note that for an
	observed variable $X$ and a set $S$ coexisting with $X$, the
	analogous relation $H(X|S) \geq 0$ already holds.)
	
	We now show that for any two coexisting sets $S,T \subset U$ with $S \cap T = \emptyset$ and where $U$ is a maximal coexisting set, the conditional entropy $H(S|T)$ is positive. First of all, by strong subadditivity,
	\begin{equation}  \label{eq:subs}
	H(S|T) \geq H(S| U \setminus S).
	\end{equation}
	Positivity of $H(S| U \setminus S)$ can be shown inductively in the cardinality of $S$. For cardinality $1$ this is implied by~\eqref{eq:subs} and~\eqref{eq:onevar}. Assuming that this holds for any set with cardinality $q$, the following shows that it also holds for any set with cardinality $q+1$. Let there be a set of variables $S \subseteq U$ of a maximal coexisting set $U$ with cardinality $q+1$ and a one element subset $S_1 \subseteq S$,  
	Writing $S=S_1\bar{S}_1$, then
	\begin{align*} 
	H(S|U\setminus S) &=H(S_1|U\setminus S)+H(\bar{S}_1|(U \setminus S) \cup S_1) \\
	&=H(S_1|U \setminus S)+H(\bar{S}_1|U \setminus \bar{S}_1),
	\end{align*} 
	which is at least 0 by the inductive hypothesis. It then follows from
	\eqref{eq:subs} that $H(S|T) \geq 0$ for all $T \subseteq (U \setminus
	S)$.

	\section{Remarks on the quantum entropy vector method} \label{app:quantum_entropy_vectors}
	\label{app:quantum}
	
	This appendix gives additional information regarding the role of Proposition~2 for the quantum entropy vector method introduced in~\cite{Chaves2015} (see also~\cite{review} for a review). For completeness, we first briefly introduce the details of this method.
	
	The \emph{quantum entropy vector method} is based on the von Neumann
	entropy. For any joint state of coexisting systems associated with
	some of the nodes (and edges) of a causal structure a joint entropy
	can be defined, where the notion of coexisting sets is the one
	discussed for the measurement entropy in the main text. However, the
	quantum method does not take the conditional entropies as separate variables (these would be redundant because $H(X|Y)=H(XY)-H(Y)$). For all variables within a coexisting set, the following inequalities hold,\footnote{In both inequalities $X$, $Y$, $Z$ are all coexisting and may each be made up of subsystems associated with several nodes.}
	\begin{itemize}
		\item Strong subadditivity: For any state $\rho_{XYZ}$, $H(XYZ)+H(Z) \leq H(XZ)+H(YZ)$.
		\item Weak monotonicity: For any state $\rho_{XYZ}$, $H(X|Y)+H(X|Z) \geq 0$.
	\end{itemize}
	Note that whenever there is no entanglement between two subsystems $X$
	and $Y$ of a state $\rho_{XY}$, the stronger monotonicity statement
	$H(X|Y)\geq0$ holds. Since it is always possible to purify an
	unobserved quantum state $\rho_{A_1 \cdots A_n}$, we can impose the
	following.
	\begin{itemize}
		\item Purification for unobserved systems: For an unobserved system in
		state $\rho_{A_1 \cdots A_n}$ we can take $H(A_1 \cdots A_n)=0$ and
		for any subsystem $S \subset \left\{A_1, \ldots, A_n \right\}$ we
		can take
		$H(S)={H( \left\{A_1, \ldots, A_n \right\} \setminus S)}$.\footnote{This can be seen by considering the Schmidt
			decomposition of the purified state~\cite{review}.}
	\end{itemize}
	Among the variables of different coexisting sets data processing inequalities hold. 
	\begin{itemize}
		\item Data Processing: Let $\rho_{X Y}$ be the joint state of two 
		sets of coexisting nodes $X$ and $Y$ and let $\mathcal{E}$ be a
		completely positive trace preserving map taking $Y$ to $Z$ such
		that $(\mathcal{I} \otimes \mathcal{E}) (\rho_{XY})=\rho_{XZ}$, then
		$H(X|Y) \leq H(X|Z)$.\footnote{For a discussion on which data processing inequalities are  relevant for computing constraints on entropy vectors we refer to~\cite{review}.}
	\end{itemize}
	In addition, the causal structure will in general imply independence constraints among observed as well as among unobserved systems. These are based on the notion of \emph{d-separation}: for three pairwise disjoint sets of variables $X$, $Y$ and $Z$, $X$ and $Y$ are d-separated by $Z$, if $Z$ blocks any path from any node in $X$ to any node in $Y$. A path is blocked by $Z$, if it contains one of the following: $i \rightarrow z  \rightarrow j$ or $i \leftarrow z \rightarrow j$ for   nodes $i$, $j$ and a node $z \in Z$ in that path, or if it contains $i \rightarrow k \leftarrow j$, where $k \notin Z$. Note that it is possible that $Z=\emptyset$.
	\begin{itemize}
		\item Independences (following Theorem~22 (i) from~\cite{Henson2014}): For three pairwise disjoint sets of observed variables, $X$, $Y$ and $Z$, if $X$ and $Y$ are d-separated by $Z$,
		then $H(X|Z)=H(X|YZ)$. (Note that $Z=\emptyset$ is allowed.)
	\end{itemize}
	
	We show with the following Lemma that weak monotonicity constraints
	are not relevant in this approach when considering causal structures
	where none of the unobserved nodes have any parents. These are the scenarios that are usually considered in the literature.
	
	\begin{lem} \label{lem:weak-mono}
		For any causal structure $C^\rQ$ in which the unobserved quantum nodes do not have any parents, all weak monotonicity inequalities are implied by the other inequalities, i.e., for any two coexisting sets $S_1, S_2$ with $S_1 \cap S_2 \neq \emptyset$,
		\begin{align}
		H(S_1 \cap S_2 | S_1 \setminus S_2) + H(S_1 \cap S_2 | S_2 \setminus S_1) \geq 0 
		\end{align}
		is redundant. 
	\end{lem}
	
	\begin{proof}
		Let $A$ denote the collection of subsystems of all unobserved nodes
		and $S$ be the maximal coexisting set that includes all unobserved
		systems $A$. 
		Then for the coexisting sets $S_1, S_2$ with
		$S_1 \cap S_2 \neq \emptyset$ we divide into three cases.
		
		\noindent
		\textit{Case~1:} $S_1, S_2 \subseteq S$.  Let $R_1=S_1 \cap A$ and $R_2=S_2 \cap A$. We use the purification of unobserved systems to rewrite,\footnote{As the unobserved nodes do not have any parents the systems associated with a particular node are independent of those associated with another node as well as independent of other variables coexisting with them.}
		\begin{align}
		H&(S_1 \cap S_2 | S_1 \setminus S_2) \! + \! H(S_1 \cap S_2 | S_2 \setminus S_1) \nonumber \\
		&\! \! \! \! \! \! \! {\small =} \! \! H(R_1) \! - \! H(R_1 \! \setminus \! R_2) \! + \! H(S_1 \! \setminus \! R_1) \! - \! H( (S_1 \! \setminus \! R_1) \! \! \setminus \! S_2) \! + \! H(R_2) \! - \! H( R_2 \! \setminus \! R_1) \! + \! H(S_2 \! \setminus \! R_2) \! - \! H( (S_2 \! \setminus \! R_2) \! \! \setminus \! S_1)\nonumber \\
		&\! \! \! \! \! \! \! {\small =}  \! \! H(R_1) \! \! - \! \! H(R_1 \! \setminus \! R_2) \! \! + \! \! H(A \! \setminus \! R_2) \! \! - \! \! H(A \! \setminus \! \! ( R_2 \! \setminus \! R_1)) \! \! + \! \! H((S_1 \! \cap \!
		S_2)\! \! \setminus\! A | (S_1 \!\setminus \! S_2) \! \! \setminus \! A) \! \! + \! \!H((S_1 \! \cap \!	S_2)\!\setminus\! \! A|(S_2 \! \setminus \! S_1) \! \! \setminus \! A) , \label{eq:str}
		\end{align}
		where to obtain the first equality we used that $\rho_{S_1}=\rho_{R_1}
		\otimes \rho_{S_1 \setminus R_1}$ and $\rho_{S_2}=\rho_{R_2} \otimes
		\rho_{S_2 \setminus R_2}$ and $R_1 \setminus S_2=R_1 \setminus R_2$
		and $R_2 \setminus S_1=R_2 \setminus R_1$; in the second line we used
		the purity of $\rho_A$.
		The last two terms in~\eqref{eq:str} are positive because these are
		classical conditional entropies (none of the sets contain elements of
		$A$).  The sum of the remaining four terms is then positive by strong
		subadditivity. 
		
		\noindent
		\textit{Case~2: $S_1 \cap S_2 \subseteq S$ and either $S_1 \not\subseteq S$
			or $S_2 \not\subseteq S$, or both.}  In this case we use data-processing to give
		\begin{align}
		H(S_1 \cap S_2 | S_1 \setminus S_2) + H(S_1 \cap S_2 | S_2 \setminus S_1) \geq  H(S_1 \cap S_2 | T_1) + H(S_1 \cap S_2 | T_2),
		\end{align}
		where $T_1,T_2 \subseteq S$ are the sets of variables that are
		processed to $S_1 \setminus S_2$ and $S_2 \setminus S_1$ respectively.
		Since $S_1 \cap S_2\subseteq S$, $S_1 \cap S_2=T_1\cap T_2$ and so
		positivity of the remaining expression follows using Case~1.
		
		\noindent
		\textit{Case~3:} $S_1 \cap S_2 \not\subseteq S$. In this case we can find $R_1 \subseteq S$ and $R_2$ with $R_2 \cap S=\emptyset$ such that $S_1 \cap S_2= R_1 \cup R_2$, and rewrite 
		\begin{align}
		H(S_1 \! \cap \! S_2 | S_1\! \setminus\! S_2) \!  + \!  H(S_1 \!  \cap \!  S_2 | S_2 \!\setminus\! S_1) 
		\! = \!  H(R_1 | S_1 \!\setminus\! S_2) \! + \! H(R_2 | S_1 \!\setminus\! R_2) \!  + \!  H(R_1 | S_2 \!\setminus\! S_1) \! + \! H(R_2 | S_2\!\setminus\! R_2)\,.
		\end{align} 
		The second and fourth terms are positive since $R_2$ is classical. The
		first and third terms correspond to a weak monotonicity inequality
		like that considered in case Case~2, and so their sum is also
		positive.
	\end{proof}
	
	By Proposition~2, instead of purifying the unobserved
	systems and dropping weak monotonicity, we could alternatively replace
	all weak monotonicity constraints by monotonicity (doing so prevents
	us from purifying the unobserved quantum systems). The question then
	arises as to the implications of each for deriving new entropy
	inequalities for the observed variables. The following lemma shows
	that the quantum approach outlined in this section (which takes the
	purification of unobserved systems into account) leads to entropy
	inequalities that are at least as tight as those obtained by
	considering monotonicity instead.

	Constraints on the observed variables are usually derived starting
	from:
	\begin{enumerate}[(1)]
		\item The Shannon constraints for the observed variables.
		\item All independences among observed and unobserved variables that
		$C$ implies.
		\item Data processing inequalities.
		\item Positivity of all entropies.
		\item Positivity of the conditional mutual information (strong subadditivity).
		\item Positivity of the conditional entropy (monotonicity) between
		subsets that cannot be entangled.
		\item \label{it:in1} Weak monotonicity between subsets that can be
		entangled.
	\end{enumerate}
	Proposition~2 implies that we can add\medskip
	
	\noindent(8) Monotonicity between subsets that can be
	entangled.\medskip
	
	\noindent and Lemma~\ref{lem:weak-mono} implies that in the case where all the
	quantum nodes are parentless, instead of (8) we can add\medskip
	
	\noindent($8'$) The unobserved systems originating at a node are
	in a pure state, e.g., for a node $A$ with subsystems $A_1, A_2$
	they obey $H(A_1, A_2)=0$ and $H(A_1)=H(A_2)$.
	
	\begin{lem}\label{lem:different_quant}
		Let $C^\rQ$ be a causal structure in which each unobserved node is
		parentless and has at most two children.  Consider starting with the
		constraints (1)--(\ref{it:in1}) and performing a Fourier-Motzkin
		elimination to give a set of constraints on the entropies of
		observed variables.  Call the resulting cone $\Gamma$.  Consider
		also $\Gamma_2$ formed analogously but with (8) in addition to
		(1)--(\ref{it:in1}) and likewise $\Gamma_2'$ with ($8'$) in addition
		to (1)--(\ref{it:in1}). We have
		$\Gamma\subseteq\Gamma_2$ and $\Gamma\subseteq\Gamma_2'$.
	\end{lem}
	
	In other words, including either (8) or ($8'$) does not give a tighter
	approximation on the set of achievable entropy vectors of the observed
	variables.

	\begin{proof}
		For any unobserved node $A$ in $C$, the subsystems $A_1$ and $A_2$ only occur jointly in coexisting sets whose state can be written as $\rho_{A_1A_2} \otimes \rho_{R}$ where $R$ contains all other systems in that coexisting set. This implies that for any $R_1 \subseteq R$,
		\begin{align} \label{eq:ind}
		H(A_1 A_2 R_1) &=H(A_1 A_2)+H(R_1) \\ 
		H(A_1 R_1) &=H(A_1)+H(R_1) \\ 
		H(A_2 R_1) &=H(A_2)+H(R_1). \label{eq:ind3}
		\end{align}
		and hence that $H(A_1 A_2 R_1)$, $H(A_1 R_1)$ and $H(A_2 R_1)$ can be
		eliminated from all valid inequalities. Strong-subadditivity and
		monotonicity inequalities that include any of these three are
		redundant since they decompose into terms that only involve $A_1$ and
		$A_2$ and terms that do not involve those variables, both of which are
		separately implied by another valid inequality.\footnote{Since
			monotonicity is only included for cq-states, the reduction does not
			lead to any decompositions that would require positivity of
			$H(A_1|A_2)$ or $H(A_2|A_1)$ to be implied.}

		The remaining types of inequalities involving both $A_1$ and $A_2$ are
		those with the form (up to exchange of $A_1$ and $A_2$)
		\begin{align}
		H(A_1)+H(A_2) &\geq H(A_1 A_2) \label{eq:nr2}\\
		H(T_1|A_1 S_2 R_2') &\geq H(T_1| A_1 A_2 R_2)  \label{eq:nr3} \\ 
		H(A_1 T_1|S_2 R_2') &\geq H(A_1 T_1| A_2 R_2)\,, \label{eq:nr4} 
		\end{align}
		where $R_1, R_2, T_1 \subseteq R$, $S_2$ is obtained by processing
		$A_2$ and $R_2$, the set $T_1$ coexists with $S_2$, and $R_2'$ is the
		subset of all observed variables in
		$R_2$. 
		
		By~\eqref{eq:ind} and~\eqref{eq:ind3}, the inequalities of types~\eqref{eq:nr3} and~\eqref{eq:nr4} are equivalent to
		\begin{align}
		H(T_1|A_1 S_2 R_2') &\geq H(T_1| R_2), \label{eq:nr5a} \\ 
		H(A_1 T_1|S_2 R_2') &\geq H(T_1| R_2)+ H(A_1A_2) - H(A_2) \label{eq:nr5}.
		\end{align}  
		For (8) we have the additional inequalities 
		\begin{align}
		H(A_1 A_2) &\geq H(A_1) \label{eq:nr1a} \\
		H(A_1 A_2) &\geq H(A_2) \label{eq:nr1b} 
		\end{align}
		and for ($8'$),
		\begin{align}
		H(A_1 A_2) &=0 \label{eq:n1} \\
		H(A_1) &= H(A_2). \label{eq:n2}
		\end{align}
		
		In the following we show that neither~\eqref{eq:nr1a} and
		\eqref{eq:nr1b} nor~\eqref{eq:n1} and~\eqref{eq:n2} imply any
		inequalities for the observed variables other than the ones that
		follow without them.

		For (8), the only remaining inequalities containing $H(A_1 A_2)$ are~\eqref{eq:nr2} and~\eqref{eq:nr5}, both of which have $H(A_1 A_2)$ as a lower bound to other entropies as well as~\eqref{eq:nr1a} and~\eqref{eq:nr1b}, where $H(A_1 A_2)$ is an upper bound. After eliminating  $H(A_1 A_2)$, we hence obtain $H(A_1)+H(A_2) \geq H(A_1)$ and $H(A_1)+H(A_2) \geq H(A_2)$ as well as $H(A_1 T_1|S_2 R_2') \geq H(T_1| R_2)$ and $H(A_1 T_1|S_2 R_2') \geq H(T_1| R_2)+H(A_1)-H(A_2)$, where the first two immediately follow from positivity of the entropy and the third is implied by~\eqref{eq:nr5a} and $H(A_1 T_1|S_2 R_2')=H(T_1|A_1 S_2 R_2')+H(A_1| S_2 R_2')$. 
		
		If we were not to impose the inequalities~\eqref{eq:nr1a} and
		\eqref{eq:nr1b}, the variable elimination would lead to
		$H(A_1)+H(A_2) \geq 0$ and
		$H(A_1 T_1|S_2 R_2') \geq H(T_1| R_2) - H(A_2)$, the first of which is
		implied by positivity and the second by~\eqref{eq:nr5a} and
		monotonicity for cq-states. We now show that after the elimination of
		$H(A_2)$ and $H(A_1)$ the additional inequalities we obtained from
		\eqref{eq:nr1a} and~\eqref{eq:nr1b},
		\begin{align}
		H(A_2) &\geq -(H(A_1 T_1|S_2 R_2')-H(T_1| R_2))+H(A_1) \label{eq:nr6} \\
		H(A_1) &\geq -(H(A_2 T_1|S_1 R_1')-H(T_1| R_1))+H(A_2) \label{eq:nr7}
		\end{align}
		become redundant (here we have put back the relation where $A_1$ and
		$A_2$ are interchanged). 
		To see this, assume that in addition the constraint $H(A_2) \geq H(A_1)$  holds (which could always be achieved by adding an independent system in a maximally mixed state to $A_2$ and which also preserves $H(A_1|A_2) \geq 0$ and $H(A_2|A_1) \geq 0$). This implies all inequalities~\eqref{eq:nr6}. Since it is then the only inequality where $H(A_2)$ upper bounds other entropies, it is in the elimination of $H(A_2)$ combined with all inequalities that involve $H(A_2)$. It furthermore renders~\eqref{eq:nr7} redundant after elimination (since only the inequalities~\eqref{eq:nr6} have $H(A_2)$ as an upper bound). Since there is no inequality left with $H(A_1)$ as an upper bound, the subsequent elimination of $H(A_1)$ leads to the same inequalities as we obtain without including  $H(A_1|A_2) \geq 0$,  $H(A_2|A_1) \geq 0$ and $H(A_2) \geq H(A_1)$.

		For ($8'$), using $H(A_1 A_2) = 0$ changes~\eqref{eq:nr2} and~\eqref{eq:nr5} to $H(A_1)+H(A_2) \geq 0$ and $H(A_1 T_1|S_2 R_2') \geq H(T_1| R_2) - H(A_2)$ respectively, which we have seen to be the inequalities that also follow upon elimination of $H(A_1 A_2)$ if~\eqref{eq:n1} is not imposed and which we have also seen to be redundant.
		
		Now, in all inequalities where $H(A_1)$ or $H(A_2)$ occur,  they are lower bounds (see~\eqref{eq:nr5a}) or as $H(A_1) \geq 0$ or $H(A_2)\geq 0$. Thus, setting $H(A_2)=H(A_1)$ and then eliminating $H(A_1)$ is equivalent to eliminating them each separately.
	\end{proof}
	
	\newpage
	\section{Inequalities for the bilocal causal structure with classical variables}
	\label{app:biloc}
	In the following we give one representative of each of the $53$ types of entropy inequalities for the bilocal causal structure. The remaining inequalities can be generated by using the symmetries under exchange of $X_0$ and $X_1$, $Y_0$ and $Y_1$ as well as $Z_0$ and $Z_1$ and the exchange of the pairs of variables $(X_0, X_1)$ with $(Z_0, Z_1)$. We list them in terms of the coefficients of the entropies in the inequalities such that each row understood as a vector ${\bf v}$ imposes an inequality ${\bf v}.{\bf H} \geq 0$ for the entropy vectors ${\bf H}$.
	
	\vspace{1cm}
	
	\hspace{-1cm}{\tiny
		$
		\begin{array}{c|ccccccccccccccc}
		\text{\#}&\! \!	H(X_0) \! \! & \! \! H(X_1) \! \! & \! \! H(Y_0) \! \! & \! \! H(Y_1) \! \! & \! \! H(Z_0) \! \! & \! \! H(Z_1) \! \! & \! \!  H(X_0 Y_0) \! \! & \! \! H(X_0 Y_1) \! \! & \! \! H(X_0 Z_0) \! \! & \! \! H(X_0 Z_1) \! \! & \! \! H(X_1 Y_0) \! \! & \! \! H(X_1 Y_1) \! \! & \! \! H(X_1 Z_0) \! \! & \! \! H(X_1 Z_1) \! \! & \! \! H(Y_0 Z_0) \\
		\hline \\
		1&0 & -1 & 0 & 0 & 0 & -1 & 0 & 0 & 0 & 0 & 0 & 0 & 0 & 1 & 0 \\
		2&0 & -1 & 0 & 0 & 0 & -1 & 0 & 0 & 0 & 0 & 0 & 0 & 0 & 0 & 0 \\
		3&0 & 0 & 0 & 0 & 0 & 0 & 0 & 0 & 0 & 0 & 0 & -1 & 0 & 0 & 0 \\
		4&0 & -1 & 0 & 0 & 0 & -1 & 0 & 0 & 0 & 0 & 0 & 0 & 0 & 0 & 0 \\
		5&0 & -1 & 0 & 0 & 0 & -1 & 0 & 0 & 0 & 0 & 0 & 0 & 0 & 0 & 0 \\
		6&0 & -1 & 0 & 0 & 0 & -1 & 0 & 0 & 0 & 0 & 0 & 0 & 0 & 0 & 0 \\
		7&0 & -1 & 0 & 0 & 0 & -1 & 0 & 0 & 0 & 0 & 0 & 0 & 0 & 0 & 0 \\
		8&0 & -1 & 0 & 0 & 0 & -1 & 0 & 0 & 0 & 0 & 0 & 0 & 0 & 0 & 0 \\
		9&0 & -1 & 0 & 0 & 0 & -1 & 0 & 0 & 0 & 0 & 0 & 0 & 0 & 0 & 0 \\
		10&0 & -1 & 0 & 0 & 0 & -1 & 0 & 0 & 0 & 0 & 0 & 0 & 0 & 0 & 0 \\
		11&0 & -1 & 0 & 0 & 0 & -1 & 0 & 0 & 0 & 0 & 0 & -1 & 0 & 0 & -1 \\
		12&0 & -1 & 0 & 0 & 0 & -1 & 0 & -1 & 0 & 0 & 0 & 1 & 0 & 0 & 0 \\
		13&0 & -1 & 0 & 0 & 0 & -1 & 0 & -1 & 0 & 0 & 0 & 1 & 0 & 0 & -1 \\
		14&0 & -1 & 0 & 0 & 0 & -1 & 0 & -1 & 0 & 0 & 0 & 1 & 0 & 0 & -1 \\
		15&0 & -1 & 0 & 0 & 0 & -1 & 0 & 0 & 0 & 0 & 0 & 0 & 0 & 0 & 0 \\
		16&0 & -1 & 0 & 0 & 0 & -1 & 0 & 0 & 0 & 0 & 0 & 0 & 0 & 0 & 0 \\
		17&0 & -1 & 0 & 0 & 0 & -1 & 0 & 0 & 0 & 0 & 0 & 0 & 0 & 0 & 0 \\
		18&0 & -1 & 0 & 0 & 0 & -1 & 0 & -1 & 0 & 0 & 0 & 1 & 0 & 0 & -2 \\
		19&0 & -1 & 0 & 0 & 0 & -1 & 0 & 0 & 0 & 0 & 0 & -1 & 0 & 0 & 0 \\
		20&0 & -1 & 0 & 0 & 0 & -1 & 0 & 0 & 0 & 0 & 0 & -1 & 0 & 0 & 0 \\
		21&0 & -1 & 0 & 0 & 0 & -1 & 0 & 0 & 0 & 0 & 0 & -1 & 0 & 0 & 0 \\
		22&0 & -1 & 0 & 0 & 0 & -1 & 0 & 0 & 0 & 0 & 0 & -1 & 0 & 0 & -1 \\
		23&0 & -1 & 0 & 0 & 0 & -1 & 0 & 0 & 0 & 0 & 0 & -1 & 0 & 0 & -1 \\
		24&0 & -1 & 0 & 0 & 0 & -1 & 0 & -1 & 0 & 0 & 0 & 1 & 0 & 0 & -1 \\
		25&0 & -1 & 0 & 0 & 0 & -1 & 0 & -1 & 0 & 0 & 0 & 1 & 0 & 0 & -1 \\
		26&0 & -1 & 0 & 0 & 0 & -1 & 0 & -1 & 0 & 0 & 0 & 1 & 0 & 0 & -1 \\
		27&0 & -1 & 0 & 0 & 0 & -1 & 0 & -1 & 0 & 0 & 0 & 1 & 0 & 0 & -1 \\
		28&0 & -1 & 0 & 0 & 0 & -1 & 0 & -1 & 0 & 0 & 0 & 0 & 0 & 0 & 0 \\
		29&0 & -1 & -1 & 0 & 0 & -1 & 0 & 0 & 0 & 0 & 0 & 0 & 0 & 0 & 0 \\
		30&0 & -1 & -1 & 0 & 0 & -1 & 0 & 0 & 0 & 0 & 0 & 0 & 0 & 0 & 0 \\
		31&0 & -1 & 0 & 0 & 0 & -1 & 0 & 0 & 0 & 0 & 0 & -1 & 0 & 0 & 0 \\
		32&0 & -1 & 0 & 0 & 0 & -1 & 0 & 0 & 0 & 0 & 0 & -1 & 0 & 0 & 0 \\
		33&0 & -1 & 0 & 0 & 0 & -1 & 0 & -1 & 0 & 0 & 0 & 1 & 0 & 0 & 0 \\
		34&0 & -1 & 0 & 0 & 0 & -1 & 0 & -1 & 0 & 0 & 0 & 1 & 0 & 0 & -2 \\
		35&0 & -1 & 0 & 0 & 0 & -1 & 0 & -1 & 0 & 0 & 0 & 1 & 0 & 0 & -2 \\
		36&0 & -1 & 0 & 0 & 0 & -1 & 0 & -1 & 0 & 0 & 0 & 0 & 0 & 0 & 0 \\
		37&0 & -1 & -1 & 0 & 0 & -1 & 0 & 0 & 0 & 0 & 0 & 0 & 0 & 0 & 0 \\
		38&0 & -1 & -1 & 0 & 0 & -1 & 0 & 0 & 0 & 0 & 0 & 0 & 0 & 0 & 0 \\
		39&0 & 0 & 0 & -1 & 0 & 0 & 0 & 0 & 0 & -1 & 0 & 0 & -1 & 0 & 0 \\
		40&0 & 0 & 0 & -1 & 0 & 0 & 0 & 0 & 0 & -1 & 0 & 0 & -1 & 0 & 0 \\
		41&0 & -1 & 0 & 0 & 0 & -1 & 0 & 0 & 0 & 0 & 0 & -1 & 0 & 0 & -1 \\
		42&0 & -1 & 0 & 0 & 0 & -1 & 0 & 0 & 0 & 0 & 0 & -1 & 0 & 0 & -1 \\
		43&0 & -1 & 0 & 0 & 0 & -1 & 0 & 0 & 0 & 0 & 0 & -1 & 0 & 0 & -1 \\
		44&0 & -1 & 0 & 0 & 0 & -1 & 0 & 0 & 0 & 0 & 0 & -1 & 0 & 0 & -1 \\
		45&0 & -1 & 0 & 0 & 0 & -1 & 0 & -1 & 0 & 0 & 0 & 0 & 0 & 0 & -1 \\
		46&0 & -1 & 0 & 0 & 0 & -1 & 0 & -1 & 0 & 0 & 0 & 0 & 0 & 0 & -1 \\
		47&0 & -1 & 0 & 0 & 0 & -1 & 0 & -1 & 0 & 0 & 0 & 0 & 0 & 0 & -1 \\
		48&0 & -1 & 0 & 0 & 0 & -1 & 0 & -1 & 0 & 0 & 0 & 0 & 0 & 0 & -1 \\
		49&0 & -1 & 0 & 0 & 0 & -1 & 0 & -1 & 0 & 0 & -1 & 1 & 0 & 0 & 0 \\
		50&0 & -1 & 0 & 0 & 0 & -1 & 0 & -1 & 0 & 0 & -1 & 1 & 0 & 0 & 0 \\
		51&0 & -1 & 0 & 0 & 0 & -1 & 0 & -1 & 0 & 0 & -1 & 1 & 0 & 0 & -1 \\
		52&0 & -1 & 0 & 0 & 0 & -1 & 0 & -1 & 0 & 0 & -1 & 1 & 0 & 0 & -1 \\
		53&0 & -1 & 0 & 0 & 0 & -1 & -1 & -1 & 0 & 0 & 1 & 0 & 0 & 0 & 0 \\
		\end{array}
		$
	}
\newpage

{\tiny
		\hspace{-1cm}$
		\begin{array}{c|ccccccccccc}
		\text{\#}&H(Y_0 Z_1) \! \! & \! \! H(Y_1 Z_0) \! \! & \! \! H(Y_1 Z_1) \! \! & \! \! H(X_0 Y_0 Z_0) \! \! & \! \! H(X_0 Y_0 Z_1) \! \! & \! \! H(X_0 Y_1 Z_0) \! \! & \! \! H(X_0 Y_1 Z_1) \! \! & \! \! H(X_1 Y_0 Z_0) \! \! & \! \! H(X_1 Y_0 Z_1) \! \! & \! \! H(X_1 Y_1 Z_0) \! \! & \! \! H(X_1 Y_1 Z_1) \\
		\hline \\
		1&0 & 0 & 0 & 0 & 0 & 0 & 0 & 0 & 0 & 0 & 0 \\
		2&0 & -1 & 1 & 0 & 0 & 0 & 0 & 0 & 0 & 1 & 0 \\
		3&0 & 0 & -1 & 0 & 0 & -1 & 1 & 0 & 0 & 1 & 1 \\
		4&0 & 0 & -1 & 0 & 0 & -1 & 1 & 0 & 0 & 1 & 1 \\
		5&0 & 0 & -1 & 0 & -1 & 0 & 1 & 0 & 1 & 0 & 1 \\
		6&0 & 0 & -1 & -1 & 0 & 0 & 1 & 1 & 0 & 0 & 1 \\
		7&0 & -1 & 0 & 0 & -1 & 1 & 0 & 0 & 1 & 0 & 1 \\
		8&0 & -1 & 0 & 0 & -1 & 0 & 1 & 0 & 1 & 1 & 0 \\
		9&0 & -1 & 0 & -1 & 0 & 1 & 0 & 1 & 0 & 0 & 1 \\
		10&0 & -1 & 0 & -1 & 0 & 0 & 1 & 1 & 0 & 1 & 0 \\
		11&1 & 0 & 0 & 0 & 0 & 0 & 0 & 0 & 0 & 1 & 1 \\
		12&0 & -1 & 1 & 0 & 0 & 1 & 0 & 0 & 0 & 0 & 0 \\
		13&0 & 0 & 0 & 1 & 0 & 0 & 0 & 0 & 1 & 0 & 0 \\
		14&0 & 0 & 0 & 0 & 1 & 0 & 0 & 1 & 0 & 0 & 0 \\
		15&0 & -2 & 1 & 0 & 0 & 1 & -1 & 0 & 0 & 1 & 1 \\
		16&0 & -2 & 1 & 0 & -1 & 1 & 0 & 0 & 1 & 1 & 0 \\
		17&0 & -2 & 1 & -1 & 0 & 1 & 0 & 1 & 0 & 1 & 0 \\
		18&1 & 0 & 0 & 1 & 0 & 0 & 0 & 1 & 0 & 0 & 0 \\
		19&0 & 0 & -1 & -1 & 0 & 0 & 1 & 0 & 1 & 1 & 1 \\
		20&0 & -1 & 1 & 0 & 0 & 1 & -1 & 0 & 0 & 1 & 1 \\
		21&0 & -1 & 0 & -1 & 0 & 1 & 0 & 0 & 1 & 1 & 1 \\
		22&1 & 0 & 0 & 1 & -1 & 0 & 0 & 0 & 0 & 1 & 1 \\
		23&1 & 0 & 0 & 0 & 0 & 0 & 0 & 1 & -1 & 1 & 1 \\
		24&0 & 0 & 0 & 1 & 0 & 1 & 0 & 0 & 1 & -1 & 0 \\
		25&0 & 0 & 0 & 1 & 0 & 0 & 1 & 0 & 1 & 0 & -1 \\
		26&0 & 0 & 0 & 0 & 1 & 1 & 0 & 1 & 0 & -1 & 0 \\
		27&0 & 0 & 0 & 0 & 1 & 0 & 1 & 1 & 0 & 0 & -1 \\
		28&0 & -1 & 0 & -1 & 0 & 1 & 1 & 0 & 1 & 1 & 0 \\
		29&0& -1 & 1 & 1 & 0 & 0 & -1 & 0 & 1 & 1 & 0 \\
		30&0 & -1 & 1 & 0 & 1 & 0 & -1 & 1 & 0 & 1 & 0 \\
		31&0 & 0 & -1 & -2 & 1 & 0 & 1 & 1 & 0 & 1 & 1 \\
		32&0 & -1 & 0 & -2 & 1 & 1 & 0 & 1 & 0 & 1 & 1 \\
		33&0 & -2 & 1 & 0 & 0 & 1 & 1 & 0 & 0 & 1 & -1 \\
		34&1 & 0 & 0 & 1 & 0 & 1 & 0 & 1 & 0 & -1 & 0 \\
		35&1 & 0 & 0 & 1 & 0 & 0 & 1 & 1 & 0 & 0 & -1 \\
		36&0 & -1 & 0 & -2 & 1 & 1 & 1 & 1 & 0 & 1 & 0 \\
		37&0 & -1 & 1 & 1 & 0 & 1 & -2 & 0 & 1 & 0 & 1 \\
		38&0 & -1 & 1 & 0 & 1 & 1 & -2 & 1 & 0 & 0 & 1 \\
		39&0 & 0 & 0 & 0 & 1 & -1 & 1 & 1 & -1 & 1 & 1 \\
		40&0 & 0 & 0 & -1 & 1 & 1 & 1 & 1 & -1 & 1 & 0 \\
		41&1 & 0 & -1 & 0 & -1 & 0 & 1 & 0 & 1 & 1 & 1 \\
		42&1 & 0 & -1 & -1 & 0 & 0 & 1 & 1 & 0 & 1 & 1 \\
		43&1 & -1 & 0 & 0 & -1 & 1 & 0 & 0 & 1 & 1 & 1 \\
		44&1 & -1 & 0 & -1 & 0 & 1 & 0 & 1 & 0 & 1 & 1 \\
		45&1 & 0 & -1 & 0 & -1 & 1 & 1 & 0 & 1 & 0 & 1 \\
		46&1 & 0 & -1 & -1 & 0 & 1 & 1 & 1 & 0 & 0 & 1 \\
		47&1 & -1 & 0 & 0 & -1 & 1 & 1 & 0 & 1 & 1 & 0 \\
		48&1 & -1 & 0 & -1 & 0 & 1 & 1 & 1 & 0 & 1 & 0 \\
		49&-1 & 0 & 0 & 0 & 1 & 0 & 1 & 1 & 1 & -1 & 0 \\
		50&-1 & -1 & 1 & 0 & 1 & 0 & 0 & 1 & 1 & 0 & 0 \\
		51&0 & 0 & 0 & 1 & 0 & 0 & 1 & 1 & 1 & -1 & 0 \\
		52&0 & -1 & 1 & 1 & 0 & 0 & 0 & 1 & 1 & 0 & 0 \\
		53&0 & 0 & -1 & 0 & 1 & 1 & 1 & -1 & 0 & 0 & 1 
		\end{array}
		$}\end{document}